\documentclass[preprint,twocolumn,10pt,amsmath,amssymb,aps,superscriptaddress,floatfix]{revtex4-2}
\usepackage{hyperref}

 \usepackage{nameref}
\usepackage[T1]{fontenc}
\usepackage[utf8]{inputenc}
\usepackage[english]{babel}
\usepackage[margin=1.5cm]{geometry}
\usepackage{float}
\usepackage{fancyhdr}
\usepackage{soul}
\usepackage{graphicx}
\usepackage{float}
\usepackage{caption}
\usepackage{subcaption}
\usepackage{amsmath, amssymb, amsfonts, amsthm, stmaryrd}
\usepackage{enumitem}
\usepackage{tikz}
\usepackage{indentfirst}
\usetikzlibrary{positioning, matrix, arrows.meta}
\usepackage{placeins}
\usepackage{csquotes}
\usepackage[toc,page,title]{appendix}
\usepackage{csquotes}

\newtheorem{theorem}{Theorem}

\newtheorem*{remark}{Remark}

\newtheorem{proposition}{Proposition}

\begin{document}
\setlength{\parindent}{1em}

\title{Sandpile Models on complex networks}

\author{Komlan Fiagbe}
\email{komlan.fiagbe@unamur.be}
\affiliation{Department of Mathematics \& Namur Institute for Complex Systems - naXys, University of Namur, Rue Joseph Grafé 2, 5000 Namur, Belgium}

\author{Jean-François de Kemmeter}
\affiliation{Department of Mathematics, Royal Military Academy, Brussels, Belgium.}

\author{Timoteo Carletti}
\affiliation{Department of Mathematics \& Namur Institute for Complex Systems - naXys, University of Namur, Rue Joseph Grafé 2, 5000 Namur, Belgium}


\begin{abstract}
We investigate the sandpile model on complex networks by developing a branching-process framework that explicitly incorporates dissipation during avalanche propagation. Unlike classical branching descriptions, which assume conservative transport and locally tree-like independence, the present approach introduces grain-loss effects directly into the offspring distribution, yielding generalized generating functions for dissipative avalanche dynamics. In the dissipative regime, avalanche-size distributions acquire exponential cutoffs while preserving topology-dependent scaling behavior. Numerical simulations confirm the theoretical predictions on sparse random networks and reveal systematic deviations in highly structured topologies. In particular, by using Holme–Kim clustered scale-free networks, we show that increasing clustering continuously lowers the avalanche exponent and enhances the probability of large cascades, demonstrating that short cycles generate strong correlations that invalidate the classical independent-branch approximation. Surprisingly, trees also exhibit substantial deviations from power-law because low edge density and the abundance of leaves constrain avalanche propagation. These results show that dissipation, clustering, and sparse connectivity fundamentally reshape avalanche size distribution of the sandpile model on networks and establish quantitative limits for branching-process descriptions of avalanche dynamics.
\end{abstract}

\maketitle

\section{Introduction}
\label{sec:intro}
The sandpile model was introduced by Bak, Tang, and Wiesenfeld \cite{Article09,Article10} (BTW) in the context of self-organized criticality (SOC). Grains are slowly added to the nodes of a lattice; when the number of grains at a node exceeds a given threshold, the node topples and redistributes grains to its neighbors, possibly triggering cascades of topplings, namely to create an {\em avalanche}. The BTW model became a fundamental framework for studying critical phenomena in statistical physics.
Initially defined on regular Euclidean lattices \cite{Article09,Article10,Article07}, the model exhibits avalanche-size distributions following power laws, a characteristic signature of criticality. These results motivated the extension of sandpile dynamics to complex networks, where heterogeneous connectivity strongly influences transport and propagation processes. Such networks naturally describe many real systems, including social, transportation, communication, biological, and neural networks~\cite{Babaalabert2002,Strogatz01Nature,BBV2008,newmanbook}.

Because we are dealing with finite size networks, to prevent indefinite accumulation of grains, the original model introduced dissipative boundary conditions allowing grains to leave the system through the lattice boundaries \cite{Article09,Article10,Article07}. Similar mechanisms were later introduced on networks through dissipation during relaxation. The effect of dissipation on avalanche statistics and stationary states has been widely studied on lattices~\cite{Lahtinen2005,Lauritsen1996SelfOrganized,bonabeau1995sandpile}. In general, dissipation drives the system away from criticality and introduces exponential cutoffs in avalanche distributions.
Avalanche propagation has also been investigated on more realistic network structures, including small-world networks \cite{Lahtinen2005}, random networks \cite{bonabeau1995sandpile}, and Bethe lattices \cite{DharMajumdar1990}. One of the main analytical approaches is based on branching-process approximations, where avalanches are modeled as branching trees generated by successive topplings. Using this framework, Goh et al. \cite{goh2003sandpileSF,lee2004branching,lee2004sandpileSF} studied sandpile dynamics on scale-free networks and showed that avalanche statistics depend explicitly on the exponent $\gamma$ of the probability for a generic node to have degree $k$, $p_d(k)\sim 1/k^\gamma$. Their results highlighted the important role of hubs and demonstrated that heterogeneous topologies can modify the universality class of the dynamics. However, these studies neglected dissipation in the branching-process formulation.

The present work extends these previous studies in several directions. First, we develop a branching-process description of sandpile dynamics on complex networks that explicitly incorporates dissipation into the offspring distribution. This leads to generalized analytical formula for avalanche-size distributions in the dissipative regime. Second, we derive explicit expressions for avalanche statistics on several classes of networks, including scale-free, Erdős–Rényi, and random regular graphs, by using generating functions and asymptotic methods. In particular in the case of scale-free networks we prove that dissipation modifies the exponent of the power law distribution of avalanche sizes. Finally, we investigate the impact of small cycles and clustering on avalanche distribution. While classical branching-process approaches assume locally tree-like structures, our numerical results on clustered scale-free networks show that short cycles significantly modify avalanche statistics and increase the probability of large cascades. We also study avalanche dynamics on lattices, random regular graphs, Cayley trees, and Barabási–Albert trees, showing that tree-like structures do not always follow the ideal power-law behavior predicted by standard branching theories.

The work is organized as follows. In Section~\ref{sec1} we introduce the sandpile model on network and the main results we will apply in Section~\ref{sec:application}, after having introduced the methodology used in the numerical simulation in Section~\ref{sec:nummeth}. Before to conclude in Section~\ref{sec:conclu}, we study of the impact of cycles in Section~\ref{sec:cycles}.

\section{Sandpile Model on a Network}
\label{sec1}
\indent Let us now present the rules of the BTW sandpile model defined on a network~\cite{goh2003sandpileSF,bonabeau1995sandpile,lee2004sandpileSF}. Consider thus a simple, undirected and connected network of $N$ nodes. At each discrete time step, one grain is added to a node chosen uniformly at random. Thus, node $i$ receives a grain with probability $1/N$, and its height $h_i$, i.e., the number of grains currently on node $i$, evolves as
\begin{equation*}
    h_i \longrightarrow h_i + 1 
    \quad \text{with probability } \frac{1}{N}\, .
\end{equation*}
 Let $k_i$ be the degree of node $i$, i.e., it is connected to $k_i$ nodes. We assume the maximal number of grains that node $i$ can hold is $k_i-1$. When $h_i$ reaches this threshold, the node becomes unstable and topples: it loses $k_i$ grains, each of which is transferred to one distinct neighbor. The toppling rule can be thus written as
\begin{align*}
    h_i &\longrightarrow h_i - k_i, \\
    h_j &\longrightarrow h_j + 1 
    \quad \text{for every neighbor } j \text{ of } i\, .
\end{align*}
Because the networks have finite size, to prevent an unbounded accumulation of grains in the system, we further assume that each transferred grain has a probability $0<f<1$ of being dissipated, i.e., being removed from the system. 

A single toppling may cause one or more neighboring nodes to become unstable, inducing further topplings. This may propagate through the network as a cascade of relaxations, called an {\em avalanche}. The latter continues until all nodes are again stable, i.e., $h_i < k_i$ for all $i$. In the following we are interested in characterizing how the topology of the underlying network affects the statistics of avalanche sizes, namely the total number of topplings during an avalanche.

\subsection{Branching Process Method}
We now introduce the branching process allowing to model an avalanche as a random tree generated by the successive topplings~\cite{lee2004branching}. In the stationary regime, a node of degree $k$ may have any height
\[
 h_i \in \{0,1,\dots,k-1\}\, .
\]
By assuming those values to be equally likely, the probability that a node of degree $k$ is just at the threshold of instability, i.e., that $h_i = k-1$, is simply $1/k$.

Let $q_k$ denote the probability that a node generates exactly $k$ branches in the avalanche tree, i.e., namely its toppling causes $k$ of its neighbors to topple. This happens precisely if the node has degree $k$, and
   just before receiving a grain, its height is $k-1$.
Hence,
\begin{equation}\label{SEC1EQ1}
    q_k = \frac{\mathbf{P}(k)}{k}\, ,
\end{equation}
where $\mathbf{P}(k)$ is the probability that a randomly chosen grain arrives at a node of degree $k$, the latter quantity is called {\em excess degree}~\cite{newmanbook}. If we choose an edge at random, the probability that it leads to a node of degree $k$ is proportional to $k$. Thus,
\begin{equation}\label{SEC1EQ2}
    \mathbf{P}(k) = \frac{k\, p_d(k)}{\displaystyle \sum_{j \ge 1} j\, p_d(j)}\, ,
\end{equation}
where $p_d(k)$ is the degree distribution of the network.
Combining~\eqref{SEC1EQ1} and~\eqref{SEC1EQ2} gives
\begin{equation}\label{SEC1EQ3}
    q_k = \frac{p_d(k)}{\displaystyle \sum_{j \ge 1} j\, p_d(j)}\,.
\end{equation}

Let $Q(w)$ be the probability generating function of the offspring distribution $\{q_k\}_{k \ge 0}$
\begin{equation}
    Q(w) = \sum_{k=0}^{\infty} q_k\, w^k,
    \qquad |w| \le 1\, ,
\end{equation}
and let $S$ denote the total size of an avalanche, i.e., the total number of nodes that topple in the cascade. We define the generating function of the avalanche size distribution as
\begin{equation}
    P(y) = \sum_{s=1}^{\infty} p(s)\, y^s,
    \qquad |y| \le 1\, ,
\end{equation}
where $p(s) = \mathbb{P}(S = s)$, i.e., the probability for an avalanche to have size $S=s$ for some $s\geq 1$. The latter quantities are related by the fundamental functional relation (a detailed derivation can be found in Appendix~\ref{app1}):
\begin{equation}\label{SEC1EQ4}
    P(y) = y\, Q\bigl(P(y)\bigr)\,.
\end{equation}

By studying the singular behavior of $P(y)$ near $y=1$, given the singular behavior of $Q(w)$ near $w=1$, we will determine the asymptotics behavior of $p(s)$ for large avalanche sizes $s$~\cite{lee2004sandpileSF}. The functional equation~\eqref{SEC1EQ4} can be formally inverted to return
\begin{equation}\label{SEC1EQ5}
    y = \frac{w}{Q(w)},
\end{equation}
where we set $w = P(y)$. Knowing the offspring distribution $\{q_k\}$, we can compute $Q(w)$, expand it for $w$ close to $1$, and deduce the corresponding expansion of $P(y)$ as $y \uparrow 1$. Finally, the coefficients $p(s)$ can be extracted via Cauchy's integral formula,
\begin{equation}\label{SEC1EQ6}
     p(s) = \frac{1}{2\pi i} \oint_{C} P(y)\, y^{-s-1}\, dy,
\end{equation}
where $C$ is a contour encircling the origin $y=0$ but not crossing the branch cut $[1,\infty)$.

\begin{remark}
The branching representation crucially relies on the assumption that the subtrees generated by different children of a node are independent and identically distributed. This is a strong simplification and is in general not exactly satisfied, even on tree-like graphs (networks without cycles). Nevertheless, the approximation is accurate for 'sufficiently random' networks~\cite{bonabeau1995sandpile}, in particular those that do not contain too many short cycles like random networks that are locally tree-like.
\end{remark}

Let us observe that Eq.~\eqref{SEC1EQ4} has been obtained without considering dissipation, we now thus fill this gap by taking into account the probability $f$ to lose one toppling grain at each transfer, and thus with probability $1-f$ the grain is successfully transmitted. The probability that a node generates $k$ branches depends now on $f$ and reads
\begin{align}
\label{SEC1EQ7}
    q_k(f)=\sum_{m=k}^{\infty}q_m\binom{m}{k}(1-f)^{k}f^{m-k} ,
\qquad k \ge 1
\end{align}
and 
\begin{align*}
q_0(f)
=1+\sum_{k=1}^{\infty}
q_k(
f^{k}-1)\, ,
\end{align*}
where $q_m$ denotes again the probability that a node generates $k$ branches in the avalanche tree~\eqref{SEC1EQ3}.

We can thus compute the new probability generating function of the offspring distribution, depending on $f$ 
\begin{align}
    \label{eq:Qwffun}
    Q(w|f)=\sum_{k=0}^{\infty} q_k(f)w^k\, 
\end{align}
and by using the definition~\eqref{SEC1EQ7} one can directly obtain
\begin{align*}
 Q(w|f)=Q((1-f)w+f)\, .
\end{align*} 
See Appendix~\ref{app2} for more details and some examples.

The functional relation relating $Q(w|f)$ and the generating function of the avalanche size distribution, can now be written as
\begin{equation}
\label{eq:PQf}
   P(y)
   = y\, Q\bigl(P(y)|f\bigr)\, ,
\end{equation}
which is a straightforward generalization of Eq.~\eqref{SEC1EQ4}.

Let us observe that when we neglect dissipation, the average number of offsprings produced by a toppling event, is given by
\begin{align*}
    \mu=\sum_{k \ge 0}k q_k=1\, ,
\end{align*}
the process is thus in a critical regime (each individual produces one offspring on average). However once $f>0$ we get (See Appendix~\ref{app2} for more details)
\begin{align}
\label{eq:muf}
    \mu(f)=\sum_{k \ge 0}k q_k(f)=1-f<1\, .
\end{align}
We can thus conclude that we are in a subcritical regime that suggests the cascades to die out exponentially fast because each individual produces less than one offspring on average.

The following Theorems~\ref{SEC1TH1} and ~\ref{thm:lagrange} (see~\cite{flajolet2009analytic}) propose alternative ways, with respect to the one above presented, to determine $p(s)$; they are rooted on a direct analysis of the Taylor expansion of the function $P(y)$. In the following Sections, we will use Theorem~\ref{SEC1TH1} to prove the sought scaling law in presence of dissipation in the case of Erd\H{o}s--R\'enyi and random regular networks; the scale free framework, however, will be solved by using Theorem~\ref{thm:lagrange}  because the former does not apply.

\begin{theorem}\label{SEC1TH1}
Let $\psi(z)$ be analytic at $0$ and satisfy the implicit equation
\[
\psi(z) = z\,\phi(\psi(z)),
\]
where $\phi(u)$ is non-periodic and analytic at $0$ and can be written as power series, $\phi(u)=\sum_{n=0}^{\infty} a_n u^n$, converging in the disk of radius $R>0$. The notation $[u^n]\phi(u)$ denotes the coefficient extraction operator, i.e.,
$[u^n]\phi(u) = a_n$. Assume the following conditions to hold true:

\textbf{(H1) Regularity and positivity.}
\begin{itemize}
    \item $\phi(0) \neq 0$,
    \item $[u^n]\phi(u) \ge 0$ for all $n$ ,
    \item $\phi(u) \not\equiv \phi_0 + \phi_1 u$ (i.e., $\phi$ is not linear).
\end{itemize}

\textbf{(H2) Characteristic equation.}
Within the disk of convergence, $|u|<R$, of $\phi$, there exists a unique
positive solution $\tau$ with $0<\tau<R$ such that
\begin{equation}
\label{eq:roottau}
\phi(\tau) - \tau \phi'(\tau) = 0\, .    
\end{equation}
Let $\rho = \frac{\tau}{\phi(\tau)}$, then $\psi(z)$ is analytic at $0$, moreover
\begin{enumerate}
    \item The radius of convergence of $\psi(z)$ is $\rho>0$.
    \item $\psi(z)$ admits the following square-root singular expansion at $z=\rho$:
    \[
    \psi(z)
    =
    \tau
    -
    \gamma \sqrt{1-\frac{z}{\rho}}
    + \mathcal{O}\!\left(1-\frac{z}{\rho}\right),
    ~~
    \gamma = \sqrt{\frac{2\phi(\tau)}{\phi''(\tau)}}.
    \]
    Moreover, a full locally convergent expansion exists in powers of
    $\sqrt{1-z/\rho}$.
    
    \item The coefficients of $\psi(z)$ satisfy the asymptotic estimate
    \begin{equation}
    \label{eq:psin}
            [z^n]\,\psi(z)
    =
    \sqrt{\frac{\phi(\tau)}{2\phi''(\tau)}}
    \frac{\rho^{-n}}{\sqrt{\pi n^{3}}}
    \left(1 + \mathcal{O}\!\left(\frac{1}{n}\right)\right)\, .
    \end{equation}
\end{enumerate}
\end{theorem}

\begin{theorem}[Lagrange Inversion Theorem]
\label{thm:lagrange}
Let
\[
\phi(u)=\sum_{k\ge 0}\phi_k u^k
\]
be a power series in \(\mathbb{C}[[u]]\) with \(\phi_0\neq 0\).
Then the equation
\[
\psi(z)=z\phi(\psi(z))
\]
admits a unique solution \(\psi(z)\in\mathbb{C}[[z]]\), whose coefficients are given by
(Lagrange form)
\begin{equation}
\psi(z)=\sum_{n=1}^{\infty} \psi_n z^n,
\qquad
\psi_n=\frac{1}{n}[u^{n-1}]\phi(u)^n.
\end{equation}
\end{theorem}

\section{Numerical simulations: Methodology}
\label{sec:nummeth}
The aim of this Section is to present the methodology we used to perform and display the numerical results. In the simulations, dissipation is controlled through the probability \(f\) that a toppling grain is removed from the system during the transfer. This parameter is crucial because it determines whether the model reaches a well-defined stationary regime. If \(f\) is chosen too large, dissipation dominates the transport and avalanches are prematurely interrupted, which suppresses large events and distorts the scaling region. Conversely, if \(f\) is too small, grains accumulate faster than they can be evacuated, and the system may exhibit an overloaded (or weakly stationary) behavior where extremely large avalanches become more frequent and finite-size saturation effects appear early. For this reason, \(f\) must be tuned to the network size: larger networks can typically sustain smaller dissipation while still remaining stable, whereas smaller networks require stronger dissipation to avoid excessive loading. Finally, to observe avalanche dynamics close to the critical regime, we choose the dissipation probability $f$ sufficiently small so that the system remains near the critical regime~\eqref{eq:muf}.

To ensure good statistical sampling, we run each simulation for a number of driving steps proportional to the network size. In practice, we use a total number of steps equal to \(20\times |E|\), where \(|E|\) is the number of edges of the network. This choice is large enough to let transient effects die out, collect a large number of avalanches spanning several decades in size, and obtain stable estimates of the avalanche-size distribution and its exponent. We run the avalanche simulations $10$ times to obtain $10$ different samples of avalanche sizes. We merge the $10$ samples into a single sample. Throughout the remainder of this work, we keep the same protocol so that the resulting distributions can be compared consistently across different network topologies and values of the control parameters.

For the fitting procedure, we employ the python power-law package \cite{Article19}, which is based on the statistical framework developed in \cite{Article20,Article21}. The fitting range is determined independently for each dataset in order to ensure the validity of the power-law approximation, since empirical distributions may deviate from ideal scaling behavior outside restricted regimes. The natural lower bound of the support is $x_{\min}=1$. However, in some cases this choice leads to systematic deviations from power-law behavior in the lower range of the data. Although the lower cutoff can be estimated by using the Kolmogorov--Smirnov minimization procedure, finite-size effects may result in an insufficient number of observations in the tail when the automatically selected $x_{\min}$ is too large. Therefore, we adopt $x_{\min}=1$ whenever it is consistent with the data; otherwise, the lower bound is increased to the smallest value for which the estimated scaling exponent remains stable and the diagnostic distances indicate a close agreement between the empirical distribution
and the fitted model.
The upper bound of the fitting range is chosen manually as the largest value for which scaling behavior is observed before clear deviations from the power-law trend appear in the empirical distribution.
To further assess the validity and consistency of the fitted model, we
use the  Kolmogorov--Smirnov $KS$, Kuiper $V$ and Anderson--Darling statistics 
\cite{Kuiper1960,Stephens1970,Anderson1952,Stephens1974,Article20,Article21}. The KS statistic is
widely used in power-law analyses because it quantifies the largest
discrepancy between the observations and the theoretical distribution,
thereby providing a direct measure of goodness-of-fit. The Kuiper
statistic complements the $KS$ distance by evaluating discrepancies across
the entire support of the distribution rather than focusing only on the
largest local deviation. The Anderson--Darling statistic provides a more
stringent criterion by placing greater emphasis on the tails of the
distribution, which is particularly relevant for heavy-tailed phenomena
such as power-laws.
We report the normalized $\kappa$ coefficient implemented in the
the python power-law package \cite{Article19}. Unlike the previous
statistics, $\kappa$ is not a standard statistic that asserts the goodness-of-fit from the literature but rather a diagnostic quantity that summarizes
the average discrepancy between the empirical and fitted distributions.
Values close to unity indicate the absence of systematic overestimation or
underestimation by the fitted model. The power-law fit is therefore interpreted as consistent with the data
when the fitted exponent remains stable, the $KS$ and Kuiper distances
remain below $10^{-1}$, the normalized Anderson--Darling ($A^2$) distance remains
below $10^{-2}$, and the $\kappa$ coefficient remains within $10^{-2}$ of
unity. We use these coefficients to assess the quality and consistency of the fits associated with the other distributions encountered throughout this work. Details of how those coefficients are computed are provided in Appendix.~\ref{app7}.
\section{Applications to different networks}
\label{sec:application}
Let us consider now several applications of the above presented theory to several relevant classes of complex networks.

\subsection{Scale-Free Network}
\label{ssec:SF}
In a scale-free network, the degree distribution follows a power law
\begin{equation}
   p_d(k) \sim k^{-\gamma},
\end{equation}
with $2 < \gamma < 3$ in many empirical networks~\cite{Babaalabert1999,newmanbook}. Such networks contain hubs (nodes of very large degree), which strongly affect the propagation of avalanches.
Let us first assume that the minimum degree is $k_{\min}=1$. By using the definition~\eqref{SEC1EQ3}, we obtain
\begin{equation}\label{SEC2EQ1}
 q_k = \frac{k^{-\gamma}}{\zeta(\gamma-1)},
 \qquad k \ge 1\, ,
\end{equation}
where $\zeta(\gamma) = \sum_{n=1}^{\infty} n^{-\gamma}$ is the Riemann zeta function. The probability $q_0$ is then determined by
\begin{equation}
     q_0
     = 1 - \sum_{k=1}^{\infty} \frac{k^{-\gamma}}{\zeta(\gamma-1)} \notag = 1 - \frac{\zeta(\gamma)}{\zeta(\gamma-1)}\, .
     \label{app4eq50}
\end{equation}
Hence the generating function $Q(w)$ can be written as
\begin{equation*}
    Q(w)
      = q_0 + \sum_{k=1}^{\infty} \frac{k^{-\gamma}}{\zeta(\gamma-1)}\, w^{k}\,.
\end{equation*}
It is convenient to express the latter expression in terms of the polylogarithm function~\cite{robinson1951bose,wood1992polylogarithms}, defined for $|w| < 1$ by
\begin{equation}
\label{eq:polylog}
    \mathrm{Li}_s(w) = \sum_{k=1}^{\infty} \frac{w^k}{k^s}
    = \frac{1}{\Gamma(s)} \int_0^{\infty} \frac{y^{s-1}\,dy}{\frac{e^y}{w} - 1},
\end{equation}
where $\Gamma(\cdot)$ denotes the Gamma function. In this way one can obtain
\begin{equation}
\label{eq:Qw}
     Q(w) = q_0 + \frac{\mathrm{Li}_{\gamma}(w)}{\zeta(\gamma-1)}\, .
\end{equation}
Goh et al.~\cite{goh2003sandpileSF} have determined the asymptotic expansion of $p(s)$ for large $s$, by analyzing the singular behavior of $Q(w)$ at $w=1$. More precisely, they used the Bose-Einstein type expansion~\cite{robinson1951bose}, for $s \in \mathbb{C} \setminus \mathbb{N}$,  the standard singularity analysis~\cite{flajolet1988singularity} and the property of the Lambert $W$-function \cite{Corless1996LambertW}, to show that
\begin{equation}\label{SEC2EQ2}
p(s) \sim
\begin{cases}
a(\gamma)\, s^{-\gamma/(\gamma-1)}, & 2 < \gamma < 3, \\[1ex]
b\, s^{-3/2} \bigl[\ln s\bigr]^{-1/2}, & \gamma = 3, \\[1ex]
c(\gamma)\, s^{-3/2}, & \gamma > 3,
\end{cases}
\qquad (s \to \infty)\, ,
\end{equation}
where $b=\sqrt{\pi/6}$, $c(\gamma)=\sqrt{\zeta(\gamma-1)/(2\pi\zeta(\gamma-2))}$ and   $a(\gamma)=-((\Gamma(1-\gamma)/\zeta(\gamma-1))^{1/1-\gamma})/(\Gamma(1/(1-\gamma)))$  are positive constants. Details about this derivation are given in the appendix~\ref{app4}. 

When dissipation is taken into account the above relation~\eqref{eq:Qw} can be generalized to obtain
\begin{equation}
\label{eq:QfSF}
     Q(w|f)
=
1 - \frac{\zeta(\gamma)}{\zeta(\gamma-1)}
+ \frac{\mathrm{Li}_{\gamma}\big(u(w)\big)}{\zeta(\gamma-1)},
\end{equation}
where $u(w) = (1-f)w + f$ (see Appendix~\ref{app2}).

In this case, the previous approach~\cite{goh2003sandpileSF} cannot be applied nor Theorem~\ref{SEC1TH1} because hypothesis (H2) is not satisfied, indeed the root $\tau$ of Eq.~\eqref{eq:roottau}, where $\phi(u)=Q(u|f)$, does not belong to the convergence disk.  We will thus resort to Theorem~\ref{thm:lagrange}. The latter allows to conclude that
\begin{equation}\label{SEC3EQ1}
  p(s)=\frac{1}{s}[w^{s-1}] Q(w|f)^s.
\end{equation}
Posing
\[
A=1-\frac{\zeta(\gamma)}{\zeta(\gamma-1)}\,\,,\,
B=\frac{1}{\zeta(\gamma-1)}\text{ and }
L(w)=\mathrm{Li}_{\gamma}(u(w)) \, ,
\]
we can rewrite~\eqref{eq:QfSF} as  $Q(w|f)=A+B\,L(w)$. Thus
\begin{equation}
\label{eq:Qwfpows}
Q(w|f)^s =(A+BL(w))^s=\sum_{m=0}^{s}\binom{s}{m}A^{s-m}(B L(w))^m\, .
\end{equation}
By using the definition of polylogarithm~\eqref{eq:polylog} we get
then
\begin{eqnarray*}
L(w)^m
=
\left(\sum_{k\ge 1}\frac{u(w)^k}{k^\gamma}\right)^m &=&
\sum_{k_1\ge 1}\cdots\sum_{k_m\ge 1}
\frac{u(w)^{k_1+\cdots+k_m}}{k_1^\gamma\cdots k_m^\gamma}\\
&=&\sum_{k_1,\dots,k_m\ge 1}
\frac{u(w)^K}{k_1^\gamma\cdots k_m^\gamma}\, ,
\end{eqnarray*}
where we defined $K=k_1+\cdots+k_m$. By using the binomial expansion
\[
u(w)^K
=
\bigl((1-f)w+f\bigr)^K
=
\sum_{n=0}^{K}\binom{K}{n}(1-f)^n f^{K-n}w^n\, ,
\]
we can conclude that
\[
[w^{s-1}]u(w)^K
=
\binom{K}{s-1}(1-f)^{s-1}f^{K-s+1}\, ,
\]
valid when $K\ge s-1$.
Therefore
\[
[w^{s-1}]L(w)^m
=
\sum_{k_1,\dots,k_m\ge 1}
\frac{[w^{s-1}]u(w)^K}{k_1^\gamma\cdots k_m^\gamma}\, ,
\]
and eventually
\[
[w^{s-1}]L(w)^m
=
(1-f)^{s-1}
\sum_{k_1,\dots,k_m\ge 1}
\frac{\binom{K}{s-1}f^{K-s+1}}{k_1^\gamma\cdots k_m^\gamma}
\mathbf{1}_{\{K\ge s-1\}}\, .
\]
Now, by using Eqs.~\eqref{SEC3EQ1} and~\eqref{eq:Qwfpows}, we find 
\begin{align}
\label{eq:psSFexact}
p(s)
&=
\frac{1}{s}
\sum_{m=0}^{s}\binom{s}{m}A^{s-m}B^m
[w^{s-1}]L(w)^m\notag\\[2mm]
&=
\frac{(1-f)^{s-1}}{s}
\sum_{m=0}^{s}\binom{s}{m}A^{s-m}B^m\times\notag\\
&\qquad \sum_{k_1,\dots,k_m\ge 1}
\frac{\binom{K}{s-1}f^{K-s+1}}{k_1^\gamma\cdots k_m^\gamma}
\mathbf{1}_{\{K\ge s-1\}}.
\end{align}
A detailed asymptotic analysis of this last result is presented in Appendix~\ref{app5}. The estimate is valid in the strongly heterogeneous regime, namely for small values of the degree exponent $\gamma$ (for instance, $\gamma \in [2,3]$), where the heavy-tailed degree distribution dominates the avalanche statistics. In this regime, the main contribution to the multiple sums in Eq.~\eqref{eq:psSFexact} arises from configurations with large total degree
\(
K = k_1 + \cdots + k_m,
\)
for which the binomial and power-law factors can be estimated asymptotically. This leads to the following leading-order behavior
\begin{equation}\label{SEC3EQ2}
p(s)
\approx
\frac{1}{\zeta(\gamma-1)}
\frac{(1-f)^{\gamma-1}}{f^\gamma}
s^{-\gamma}.
 \end{equation}
Let us emphasize the change in the power-law exponent $\tau$ induced by dissipation. In the absence of dissipation ($f=0$), it has been shown that, in the regime $2<\gamma<3$, the avalanche-size distribution obeys a power law with exponent
\(
\tau=\frac{\gamma}{\gamma-1},
\)
see Eq.~\eqref{SEC2EQ2}. In contrast, when dissipation is introduced ($0<f<1$), the exponent becomes
\(
\tau=\gamma.
\)
In practice, however, simulations are performed on finite-size networks. Consequently, when $f$ is very small, the system remains close to the critical regime and the avalanche-size distribution is dominated by the scaling behavior described by Eq.~\eqref{SEC2EQ2}, rather than by the dissipative asymptotics of Eq.~\eqref{SEC3EQ2}. To observe the scaling predicted by Eq.~\eqref{SEC3EQ2}, the dissipation must therefore be sufficiently strong to drive the system away from criticality. For instance, for scale-free networks of size $N=10\,000$, our numerical experiments indicate that dissipation probabilities of at least $f\gtrsim 0.4$ are required before the dissipative regime becomes clearly visible. Numerical evidence supporting that is presented in Fig~\ref{sp_Sf_All} of Appendix.~\ref{app5}.
Figures \ref{sp_Sf} and \ref{sp_Sf_2} present the avalanche-size distributions obtained on scale-free networks in the regime of very feeble dissipation ($f=4\times10^{-3}$). In this limit, the system remains close to the critical state assumed in the branching-process analysis developed in Section~\ref{sec1} with no dissipation. The distributions exhibit an extended power-law regime over several decades, confirming the scale-invariant nature of avalanche propagation. Moreover, the fitted exponents decrease as the degree exponent $\gamma$ increases, in agreement with the theoretical prediction $\tau=\gamma/(\gamma-1)$. As shown in Fig.~\ref{sp_Sf_2}, the numerical estimates follow the same trend as the theoretical curve, although finite-size effects and residual correlations produce small quantitative deviations. Overall, these results support the validity of the branching-process description for weakly dissipative scale-free networks and confirm that network heterogeneity controls the avalanche statistics through the exponent $\gamma$.
A qualitatively different behavior emerges in the regime of strong dissipation, $f=0.5$, shown in Figs.~\ref{sp_Sf_Diss} and \ref{sp_Sf_Diss_2}. In this case, the branching process becomes subcritical since the average number of offspring satisfies $\mu(f)=1-f<1$, implying that avalanches cannot propagate indefinitely. Consequently, the avalanche-size distributions are no longer characterized by a broad critical regime and instead display a pronounced truncation of large events. The fitted exponents now increase with $\gamma$, ranging from approximately $2.08$ for $\gamma=2.1$ to $3.21$ for $\gamma=3.0$, values that are close to the theoretical prediction $\tau=\gamma$ derived for dissipative scale-free networks. The comparison between Figs.~\ref{sp_Sf_2} and \ref{sp_Sf_Diss_2} therefore illustrates the transition predicted by the theory: very feeble dissipation preserves the critical scaling $\tau=\gamma/(\gamma-1)$ associated with self-organized criticality, whereas strong dissipation drives the system into a subcritical regime where avalanche propagation is strongly suppressed and the asymptotic scaling is governed by $\tau=\gamma$. This transition highlights the fundamental role of dissipation in controlling the universality class of avalanche dynamics on scale-free networks (see also Fig.~\ref{Br_sf}). Furthermore, as predicted by our theoretical analysis, Figs.~\ref{sp_Sf_2} and \ref{sp_Sf_Diss_2} reveal a crossover behavior for values of $\gamma$ slightly larger than $3$. In the presence of strong dissipation, the avalanche-size distribution is characterized by the exponent $\tau=\gamma$, whereas for weak dissipation the observed exponent remains close to the value $\tau\simeq 3/2$.


  \begin{figure}[H]
    \centering
    \includegraphics[width=0.95\linewidth]{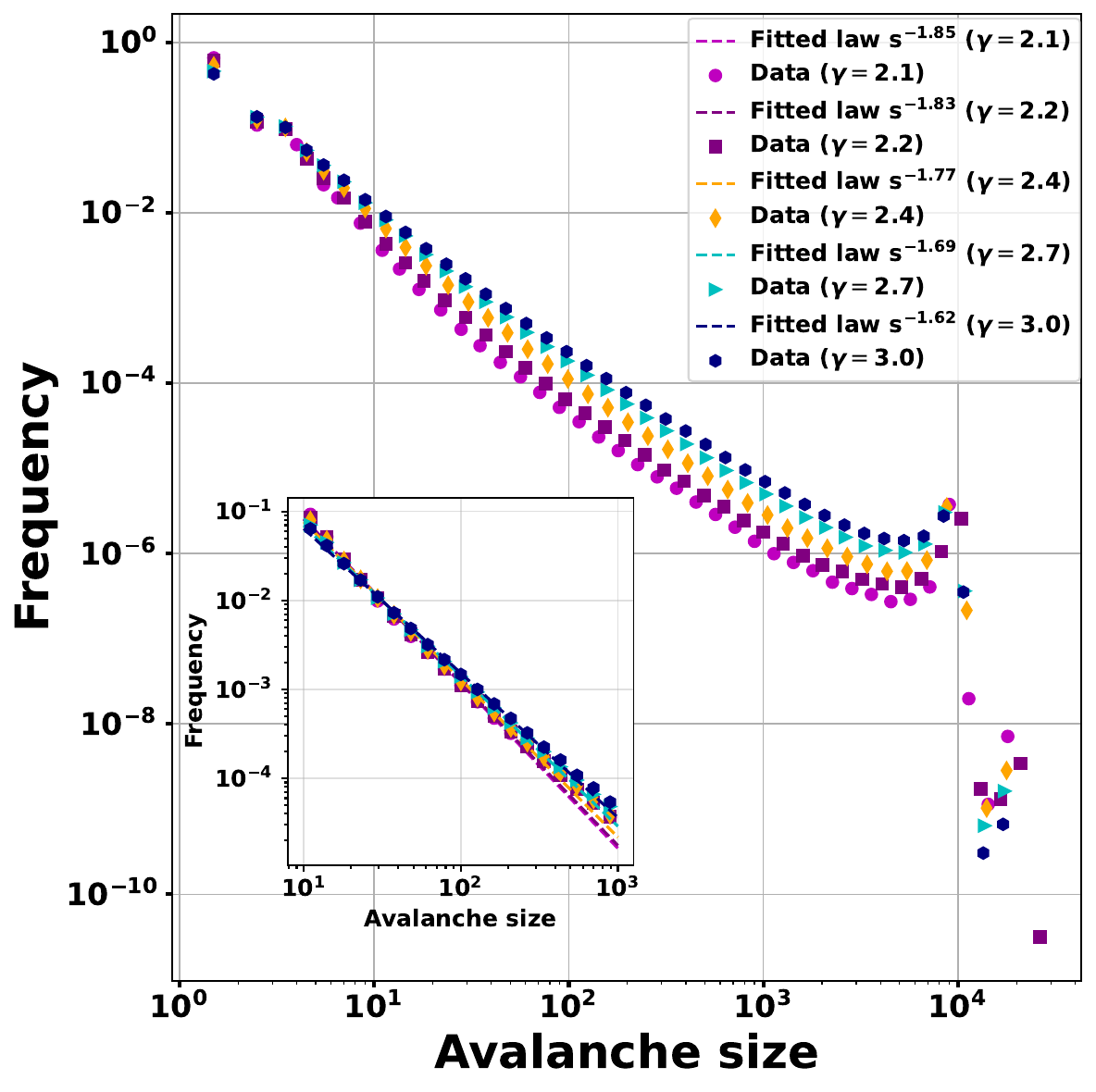}
   \caption{Main panel shows the distributions of avalanche sizes obtained with the sandpile model on different scale-free graphs. The number of nodes in each graph is $10\,000$. The dissipation probability is $4\times 10^{-3}$ in each case. The data are grouped into logarithmically spaced bins before being displayed on a log–log plot. Inset shows distributions fitted on range $[10,1000]$. The goodness-of-fit diagnostics remain moderate over the considered values of $\gamma$, with
$2.72\times10^{-2}\leq KS\leq4.52\times10^{-2}$,
$5.01\times10^{-2}\leq V\leq6.52\times10^{-2}$,
$1.00\times10^{-3}\leq A^2\leq3.59\times10^{-3}$,
and
$7.28\times10^{-3}\leq |\kappa-1|\leq8.66\times10^{-3}$.}
   \label{sp_Sf}
\end{figure}
\begin{figure}[H]
    \centering
    \includegraphics[width=0.95\linewidth]{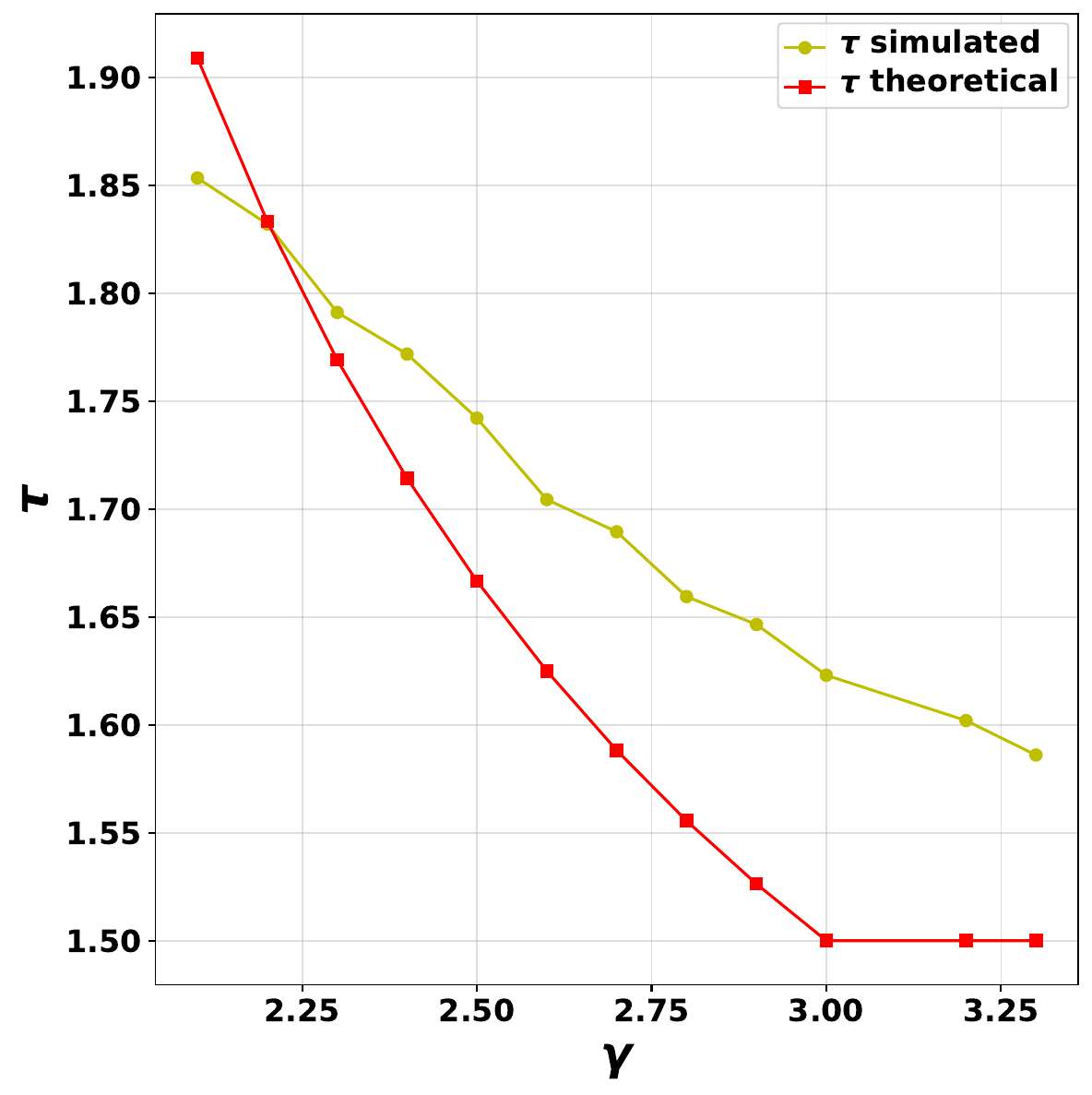}
   \caption{Variation of exponents $\tau$ we obtained by fitting the distribution by power-law with respect to $\gamma$ obtained from data shown in Figure~\ref{sp_Sf}. The dissipation probability is $f=4\times 10^{-3}$.}
   \label{sp_Sf_2}
\end{figure}

  \begin{figure}[H]
    \centering
    \includegraphics[width=0.95\linewidth]{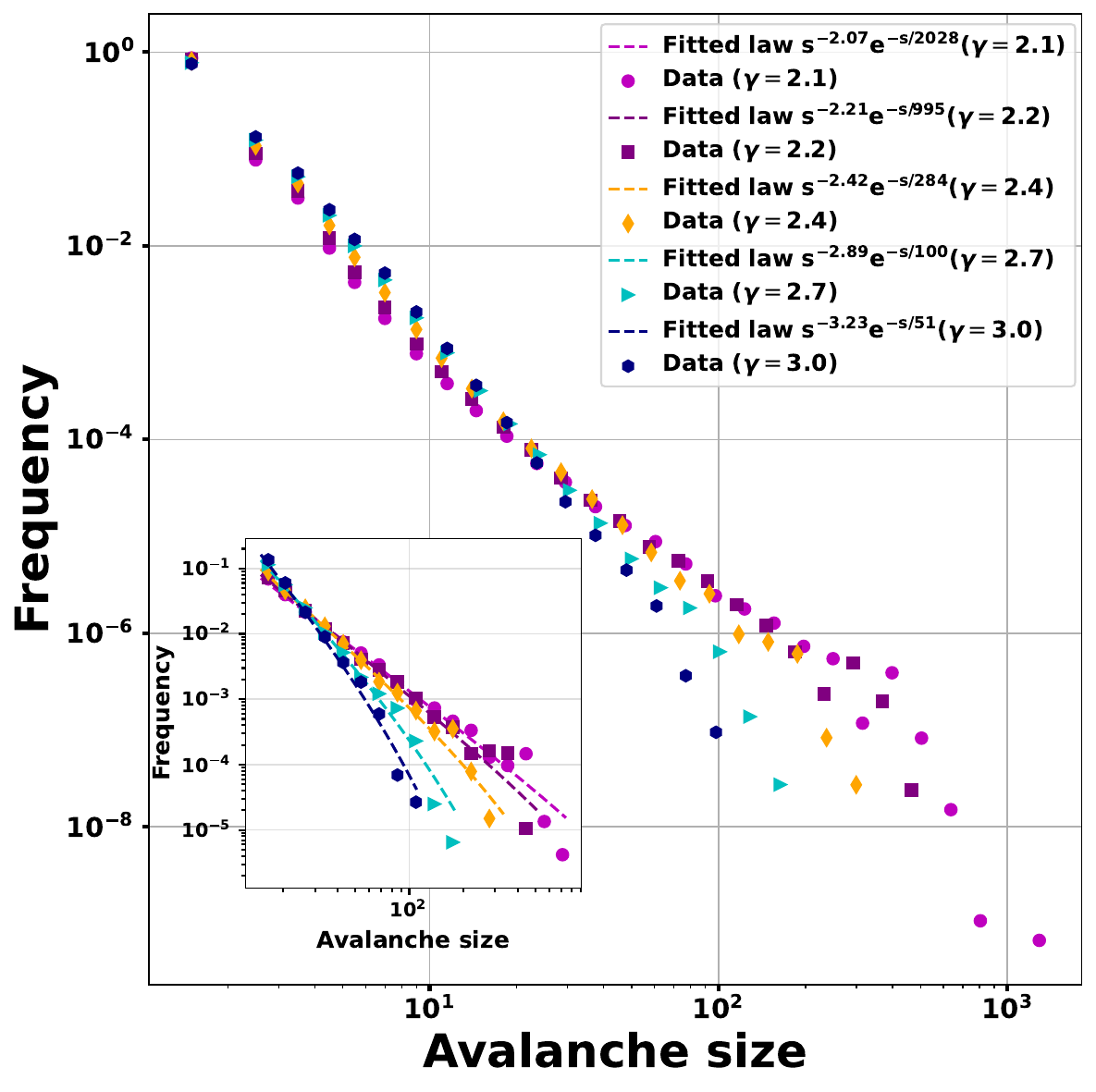}
   \caption{Main panel shows the distributions of avalanche sizes obtained with the sandpile model on different scale-free graphs. The number of nodes in each graph is $10\,000$. The dissipation probability is $0.5$ in each case. The data are grouped into logarithmically spaced bins before being displayed on a log–log plot. Inset shows distributions fitted on range $[12,1000]$. The goodness-of-fit diagnostics remain small over the considered values of $\gamma$, with
$8.43\times10^{-3}\leq KS\leq1.88\times10^{-2}$,
$1.66\times10^{-2}\leq V\leq3.19\times10^{-2}$,
$2.11\times10^{-5}\leq A^2\leq4.52\times10^{-5}$,
and
$7.46\times10^{-4}\leq |\kappa-1|\leq1.49\times10^{-3}$.}
   \label{sp_Sf_Diss}
\end{figure}
\begin{figure}[H]
    \centering
    \includegraphics[width=0.95\linewidth]{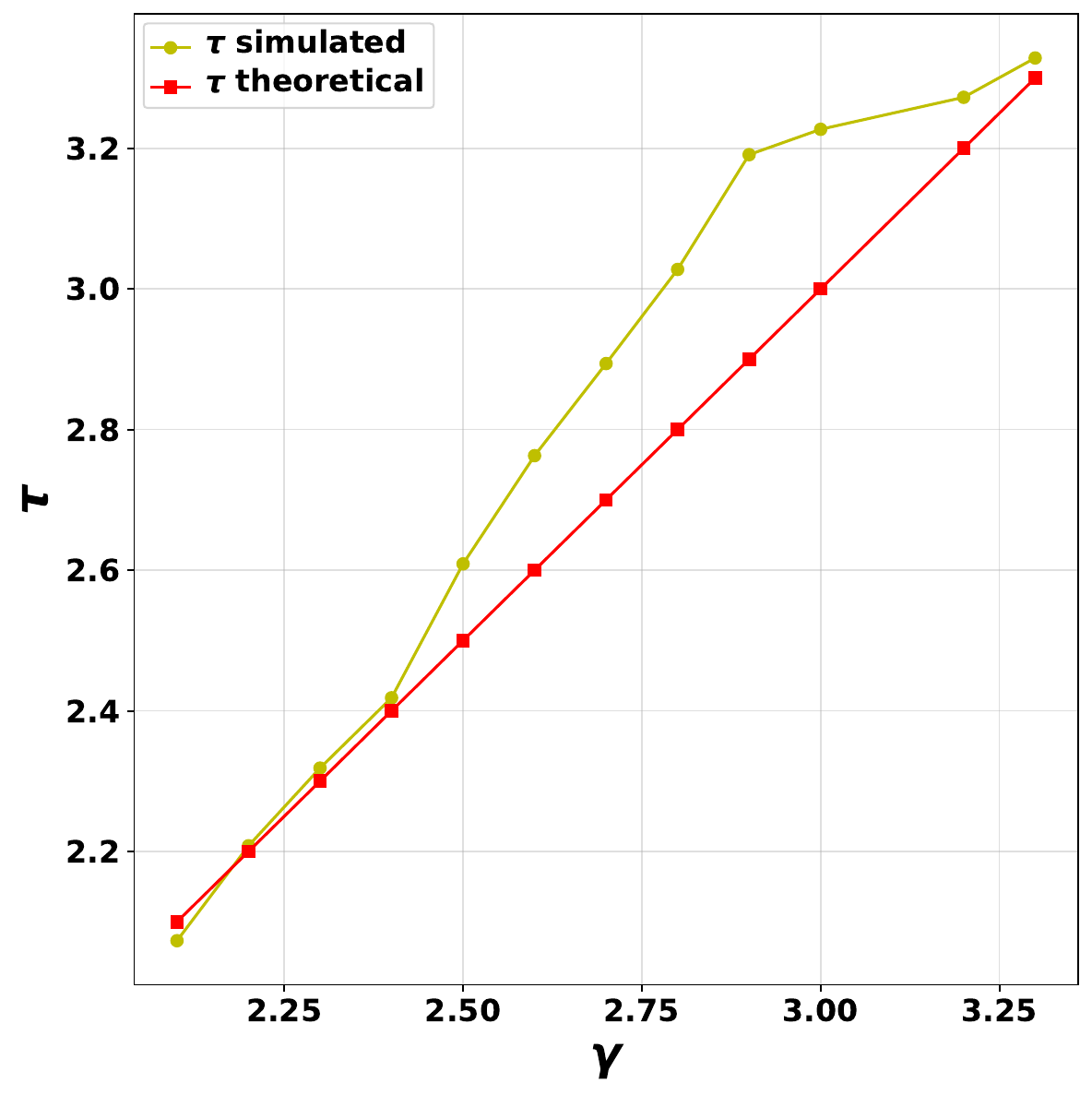}
   \caption{Variation of exponents $\tau$ we obtained by fitting the distribution by power-law with respect to $\gamma$ obtained from data shown in Figure~\ref{sp_Sf_Diss}. The dissipation probability is $f=0.5$.}
   \label{sp_Sf_Diss_2}
\end{figure}

\subsection{Erd\H{o}s--R\'enyi Network}
\label{ssec:ER}
We now consider the sandpile model defined on Erd\H{o}s-R\'enyi (ER) random graph \cite{ErdosRenyi1959} where each pair of nodes is independently connected with probability $p$. For large $N$ (number of nodes), the degree distribution is Poisson with mean $\lambda = p(N-1)$
namely
\begin{equation}
    p_d(k) = e^{-\lambda} \frac{\lambda^k}{k!}\, .
\end{equation}
Let $q_k$ be the probability that a node generates $k$ branches in the avalanche tree, then by using~\eqref{SEC1EQ3} we obtain
\begin{equation}
    q_k
     = \frac{\lambda^{k-1} e^{-\lambda}}{k!},
    \qquad k \ge 1\, ,
    \label{SEC2EQ3}    
\end{equation}
and
\begin{equation}
 q_0
 = 1 - \frac{1}{\lambda} + \frac{e^{-\lambda}}{\lambda}\, ,
 \label{SEC2EQ4}
\end{equation}
thus the generating function $Q(w)$ is 
\begin{equation}
Q(w)
 = 1 - \frac{1}{\lambda} + \frac{e^{\lambda(w-1)}}{\lambda}\, .
 \label{SEC2EQ5}
\end{equation}
By expanding $Q(w)$ in a Taylor series around $w=1$ we obtain
\begin{align*}
    Q(w)
    &= w + \frac{\lambda}{2}(w-1)^2 + {o}\bigl((w-1)^2\bigr)\, ,
\end{align*}
hence $Q''(1) = \lambda$. From Eq.~\eqref{SEC1EQ4} we obtain
\begin{equation}
    P(y) \approx 1 - \sqrt{\frac{2}{\lambda}}\,(1-y)^{1/2},
    \qquad y \uparrow 1\, ,
\end{equation}
and by using the Cauchy integral formula~\eqref{SEC1EQ6} we determine $p(s) \sim s^{-3/2}$ for large $s$. Let us observe that we could have obtained this same result from the Theorem~\ref{SEC1TH1}.

When dissipation is incorporated into the framework, we obtain
\begin{equation*}
Q(w|f)=
 1 - \frac{1}{\lambda}
   + \frac{e^{\lambda (u(w)-1)}}{\lambda}\,,
\end{equation*}
where $u(w) = (1-f)w + f$ (see Appendix~\ref{app2}).
One can easily verify that hypothesis (H1) of  Theorem~\ref{SEC1TH1} holds true. To verify the second hypothesis (H2) we look at the solution of the equation
\begin{equation*}
   Q(w^*|f)-w^*Q(w^*|f)=0\, ,
\end{equation*}
namely
\begin{align*}
   (1-w^*\lambda(1-f))\exp(\lambda((1-f)(w^*-1)))=1-\lambda\, .
\end{align*}
By posing $a=\lambda(1-f)$, the latter equation rewrites
\begin{align*}
(1-aw^*)\exp(a(w^*-1))=1-\lambda\, ,
\end{align*}
or equivalently 
\begin{align*}
(1-aw^*)\exp(aw^*)=(1-\lambda)\exp(a)\, .
\end{align*}
Let us define $y=1-aw^*$, hence
\begin{align*}
-y\,\exp(-y)=(\lambda-1)\exp(a-1 ) \, .
\end{align*}
In conclusion we obtain
\begin{align*}
y=-W\!\left((\lambda-1)\exp(a-1)\right)\, ,
\end{align*}
where $W$ is the Lambert W-function, and eventually
\begin{align}
\label{eq:wstarER}
w^*=\frac{1+W\!\left((\lambda-1)e^{\lambda(1-f)-1}\right)}{\lambda(1-f)},
\qquad \text{if } \lambda(1-f)\neq 0\, .
\end{align}
This solution is always positive for $0<f<1$; indeed $\lambda >1$ because of the assumption of connected network (see Fig~\ref{Er} for the case $\lambda=10$). 

The radius of convergence of $Q(w|f)$ at $w=0$ is $\infty$, hence hypothesis (H2) is satisfied and we can conclude
\begin{align}
\label{eq:psER}
p(s)=C \rho^{-s} s^{-3/2}\, ,
\end{align}
where $\rho=w^*/Q(w^*|f)$  and $C$ is some positive constant. 
Numerical evidence supporting this result is provided in Fig.~\ref{Sp_Br_er}. In Fig.~\ref{sp_Er} we report similar results for different values of $p$, the probability for a link to exist, and we can observe that the distributions follow the same law with exponent $3/2$, the latter results thus to be robust and independent from the connectivity $p$ and the dissipation $f$. We also provided another formulation of $p(s)$ using the Theorem~\ref{thm:lagrange} in Appendix~\ref{app6}.
\begin{figure}[ht]
    \centering
    \includegraphics[width=0.95\linewidth]{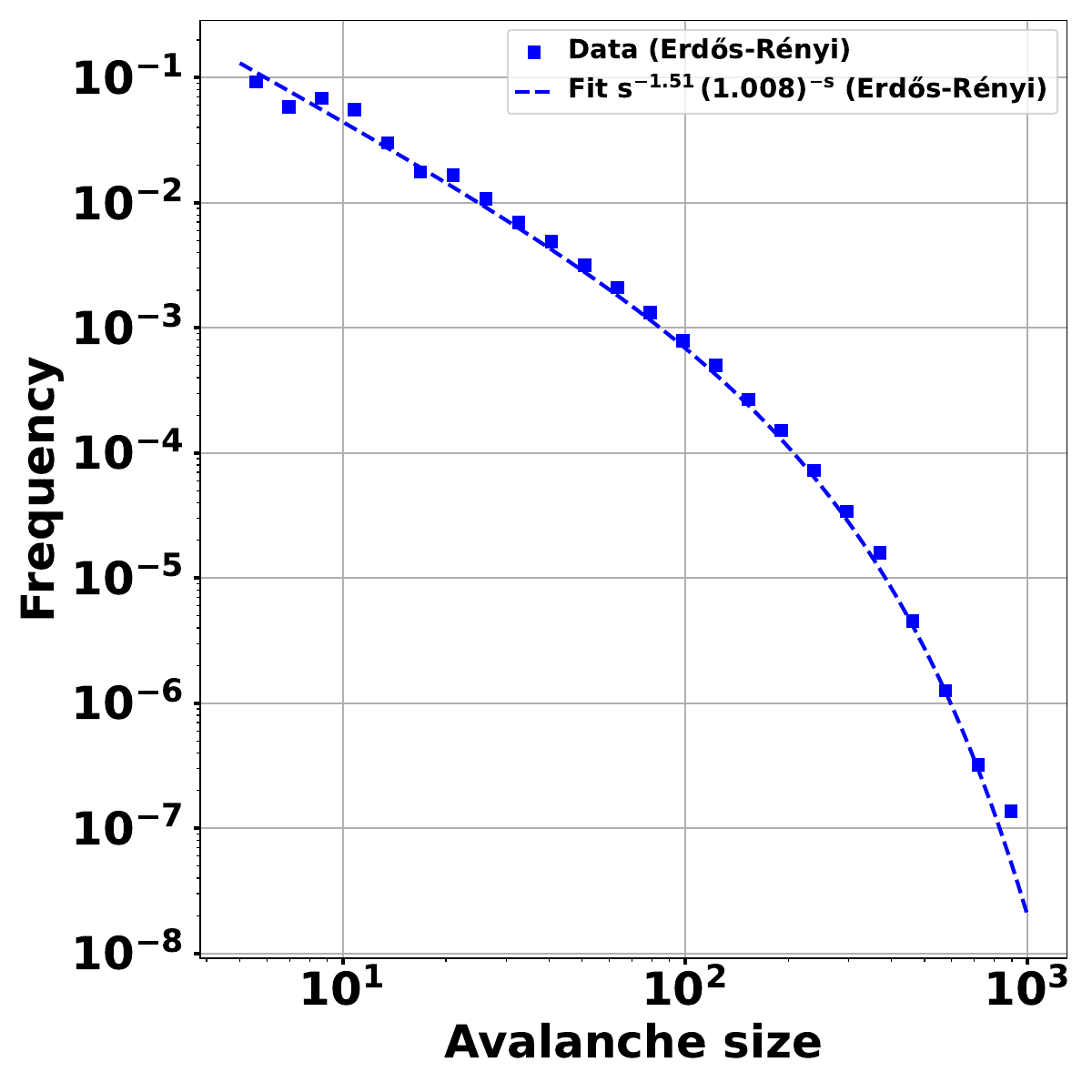}
    \caption{Distributions of avalanche sizes obtained with the sandpile model on an Erd\H{o}s--R\'enyi graph with a probability than an edge exists between any pair of nodes $p=0.001$. The number of nodes in each graph is $10\,000$. The dissipation probability is $f=0.1$ in each case. The data are grouped into logarithmically spaced bins before being displayed on a log–log plot. Distributions are fitted on range $[5,1000]$. The goodness-of-fit diagnostics remain extremely small, with
$KS=7.46\times10^{-4}$,
$V=1.43\times10^{-3}$,
$A^2=1.38\times10^{-8}$,
and
$|\kappa-1|=1.67\times10^{-5}$. }
    \label{Sp_Br_er}
\end{figure}

\subsection{Random Regular Network}
In a random $k_0$-regular graph \cite{Mckaywormald1990}, every node has the same degree $k_0$ and the degree distribution is therefore 
\begin{equation}
    p_d(k) = \delta_{k,k_0}\,,
\end{equation}
where $\delta_{k,k_0}$ is the Kronecker delta. In other words, all nodes have exactly $k_0$ neighbors. We assume $k_0 > 2$; if $k_0 = 1$, the network is not connected if there are more than $N=2$ nodes, and if $k_0 = 2$, we obtain a cycle graph rather than a genuinely random graph.

From~\eqref{SEC1EQ3}, the only non-zero element of the offspring probability $q_k$ is
\begin{equation*}
 q_{k_0} = \frac{p_d(k_0)}{\sum_j j\,p_d(j)}= \frac{1}{k_0}\, ,   
\end{equation*}
and $q_k = 0$ for all others $k \ge 1$. For $k=0$ we obtain
\begin{equation*}
   q_0 = 1 - q_{k_0}= \frac{k_0 - 1}{k_0}\, .
   \label{SEC2EQ6}
\end{equation*}
Hence the associated generating function is
\begin{equation}
    Q(w) = \frac{k_0 - 1}{k_0} + \frac{1}{k_0} w^{k_0}.
\end{equation}
By expanding around $w=1$,
\begin{equation*}
    Q(w)
    = w + \frac{k_0-1}{2}(w-1)^2
       + {o}\bigl((w-1)^2\bigr)\, ,
\end{equation*}
thus $Q''(1) = k_0 - 1$. The general critical branching-process result~\eqref{SEC1EQ4} then gives,
\begin{align}
P(y) \approx 1 -
\sqrt{\frac{2}{k_0 - 1}}\,(1-y)^{1/2},
\qquad y \uparrow 1,
\end{align}
and again $p(s) \sim s^{-3/2}$ for large avalanche sizes $s$, as one can prove by using the Cauchy integral formula~\eqref{SEC1EQ6}. 

When dissipation is incorporated into the model, the expression for the modified generating function becomes
\begin{equation*}
Q(w|f)
=
 \frac{k_0-1}{k_0}+\frac{u(w)^{k_0}}{k_0}\, ,
\end{equation*}
where $u(w) = (1-f)w + f$ (see Appendix~\ref{app2} for details).
In this case, it is straightforward to verify that the first hypothesis of Theorem~\ref{SEC1TH1} holds true. 
For the second hypothesis, we need again to determine the solution of the equation
\begin{equation*}
   Q(w^*|f)-w^*Q(w^*|f)=0\, .
\end{equation*}
By posing $y=(1-f)w^*+f$, we get 
\begin{equation}
\label{eq:polyk0}
(1-k_0)y^{k_0}+f y^{k_0-1}+k_0-1=0\, .
\end{equation}
It is difficult to find a closed form for $w^*$, however we can prove the existence of a positive root by applying the Descartes rule of signs and the observation that the leading coefficient of~\eqref{eq:polyk0} is negative. This result can be numerically confirmed (see Fig.~\ref{Rr} corresponding to $k_0=10$). Moreover, the radius of convergence of $Q(w|f)$ at $w=0$ is infinite. Consequently we can apply Theorem~\ref{SEC1TH1} and conclude that
\begin{align}
\label{eq:psregk0}
p(s)=C \rho^{-s} s^{-3/2}\, ,
\end{align}
where $\rho = \frac{w^*}{Q(w^*|f)}$ and $C>0$ is a constant. Numerical evidence supporting this result can be observed in Fig~\ref{Sp_Br_rr}. Another formulation of $p(s)$ using the Theorem~\ref{thm:lagrange} is provided in Appendix~\ref{app6}.
\begin{figure}[ht]
    \centering
    \includegraphics[width=0.95\linewidth]{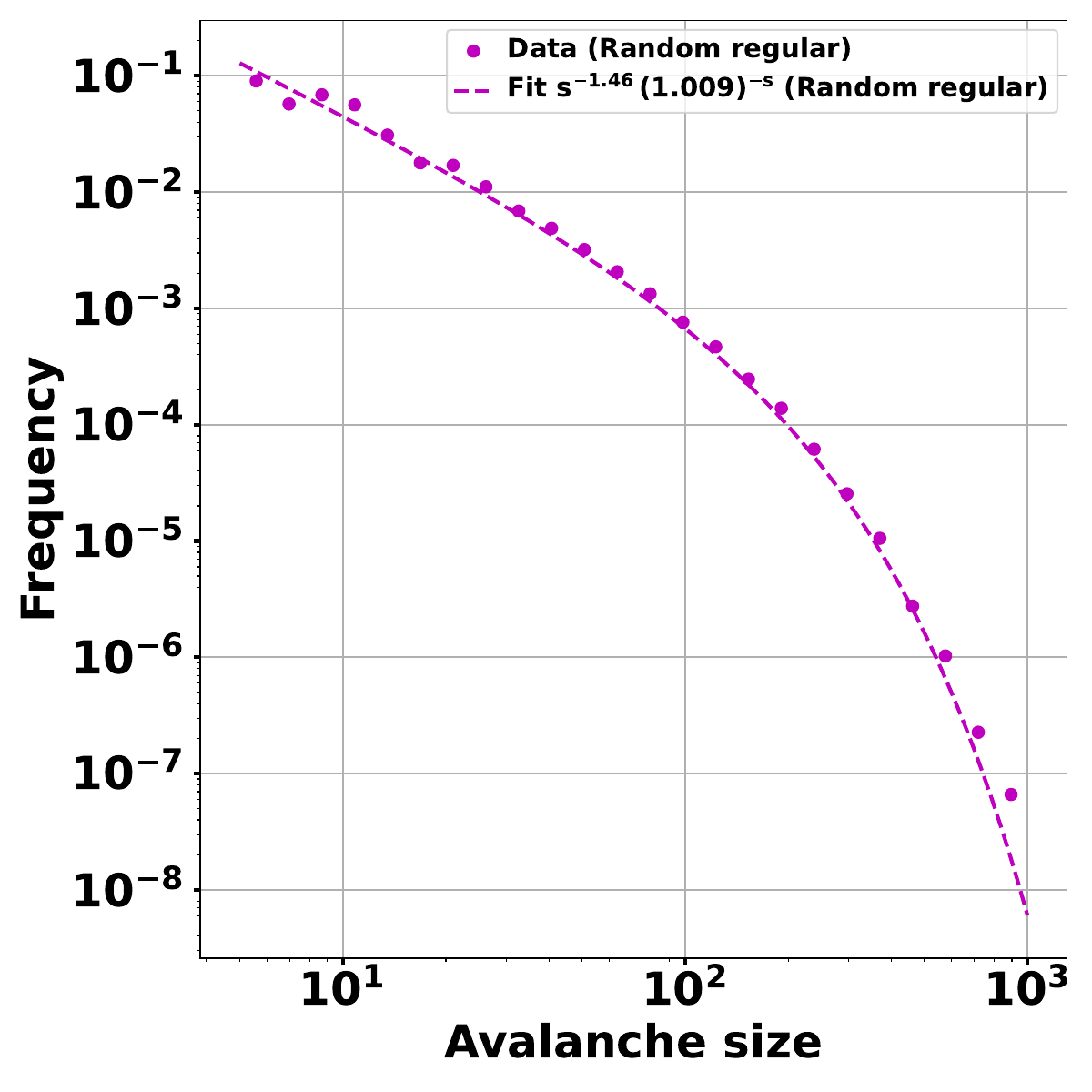}
    \caption{Distributions of avalanche sizes obtained with the sandpile model on a random $10$-regular graph. The number of nodes in each graph is $10\,000$. The dissipation probability is $f=0.1$ in each case. The data are grouped into logarithmically spaced bins before being displayed on a log–log plot. Distributions are fitted on range $[5,1000]$. The goodness-of-fit diagnostics satisfy
$KS=1.99\times10^{-3}$,
$V=3.92\times10^{-3}$,
$A^2=2.77\times10^{-8}$,
and
$|\kappa-1|=8.04\times10^{-5}$. }
    \label{Sp_Br_rr}
\end{figure}
\subsection{Comparison of asymptotic avalanche size distributions}
The study reported in the previous sections allowed us to clearly differentiate the behavior of the sandpile model defined on scale-free networks with respect to the Erdős–Rényi and random-regular topology. The aim of this section is to compare and discuss those results with the help of Figures~\ref{Br_sf},~\ref{Br_er} and~\ref{Br_rr} where we report respectively the theoretical distribution of $p(s)$ for the scale-free network given by Eq.~\eqref{SEC2EQ2} if $f=0$ and Eq.~\eqref{SEC3EQ2} if $0<f<1$, the Erdős–Rényi case Eq.~\eqref{eq:psER} and the random-regular one, Eq.~\eqref{eq:psregk0}. 

The scale-free case is special because the degree distribution itself is heavy-tailed, so large hubs can sustain avalanche propagation even when dissipation is introduced. As a result, the theoretical distributions in Fig.~\ref{Br_sf} do not show an exponential truncation when $f>0$ but a clear power law with exponent $\tau=\gamma$; instead, increasing dissipation mainly shifts the power-law behavior from the nearly conservative exponent toward the dissipative scaling predicted for scale-free networks, namely the term $(1-f)^\gamma/f^\gamma$ in Eq.~\eqref{SEC3EQ2}. Therefore, the exponential cutoffs observed in the numerical scale-free simulations should be interpreted primarily as finite-size effects: in a finite network, the maximum degree and the maximum avalanche size are bounded, which necessarily produces a cutoff at large $s$.
This behavior contrasts with Figs.\ref{Br_er} and \ref{Br_rr}, corresponding respectively to the Erdős–Rényi and random-regular cases. In these networks, the degree distribution is not heavy-tailed, and the branching-process theory predicts a distribution of the form $p(s)\sim s^{-3/2}\rho^{-s}$ when dissipation is present. Hence, the cutoff is intrinsic to the dissipative dynamics and is further reinforced by the finite size of the simulated system. In other words, for Erdős–Rényi and random-regular graphs, both dissipation and finite-size effects contribute to the exponential decay of the tail. The comparison with Fig.~\ref{Br_sf} therefore shows that scale-free networks form a distinct case: under strong dissipation, the dominant effect is not simply a truncation of the same law, but rather a shift toward a different power-law regime governed by the scale-free exponent. The exponent $\tau=3/2$ is a robust mean-field prediction of branching-process theory and arises whenever the offspring distribution has a finite second moment. In such cases, the detailed form of the offspring distribution does not affect the universality class of the avalanche dynamics. Therefore, the minimal ingredient required to deviate from this exponent is the presence of a sufficiently broad offspring distribution, typically characterized by a divergent variance. In the present framework, this situation naturally occurs in scale-free networks, where the degree distribution follows $p_d(k)\sim k^{-\gamma}$ ($2<\gamma\leq3$). The associated offspring distribution inherits this heavy-tailed character, causing the avalanche statistics to become explicitly dependent on the network heterogeneity. As a consequence, the mean-field exponent $\tau=3/2$ is no longer observed. In the weakly dissipative regime, the avalanche-size distribution is characterized by the exponent $\tau=\gamma/(\gamma-1)$, while strong dissipation leads to a different scaling behavior with $\tau=\gamma$. By contrast, Erdős--Rényi and random regular networks generate offspring distributions with finite variance and therefore retain the mean-field exponent $\tau=3/2$, with dissipation only introducing an exponential cutoff.

The distribution displayed in Fig.~\ref{Br_sf} correspond to the exact theoretical case $f=0$. By contrast, the numerical results shown in Figs.~\ref{sp_Sf} and ~\ref{sp_Sf_2} were obtained using a very small but nonzero dissipation probability to guarantee stationarity in finite systems. This weak dissipation slightly shifts the system away from the critical point and leads to fitted exponents that are systematically larger than the theoretical prediction $\tau=\gamma/(\gamma-1)$. The observed discrepancies in Fig.~\ref{sp_Sf_2} should therefore be regarded as finite-dissipation corrections. We can observe for instance the exponent shift in Fig.~\ref{Br_sf} for $f=0.012$.
\begin{figure}
    \centering
    \includegraphics[width=0.95\linewidth]{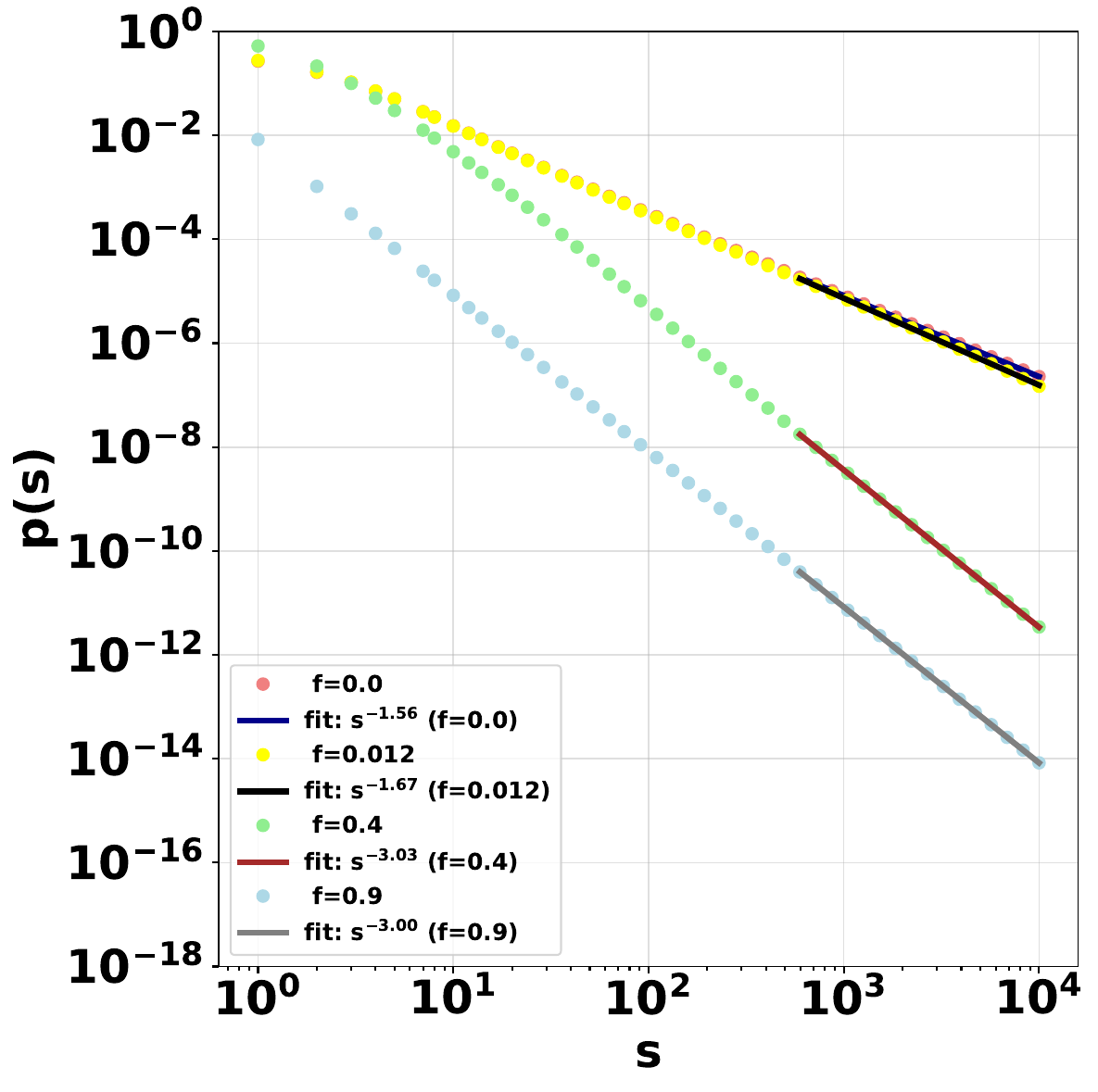}
    \caption{Avalanche size distributions for different values of the dissipation probability $f$ in the scale-free case with $\gamma=3$ as obtained from Theorem~\ref{thm:lagrange}. The  fitting range is $[5,10\,000]$. }
    \label{Br_sf}
\end{figure}
\begin{figure}
    \centering
    \includegraphics[width=0.95\linewidth]{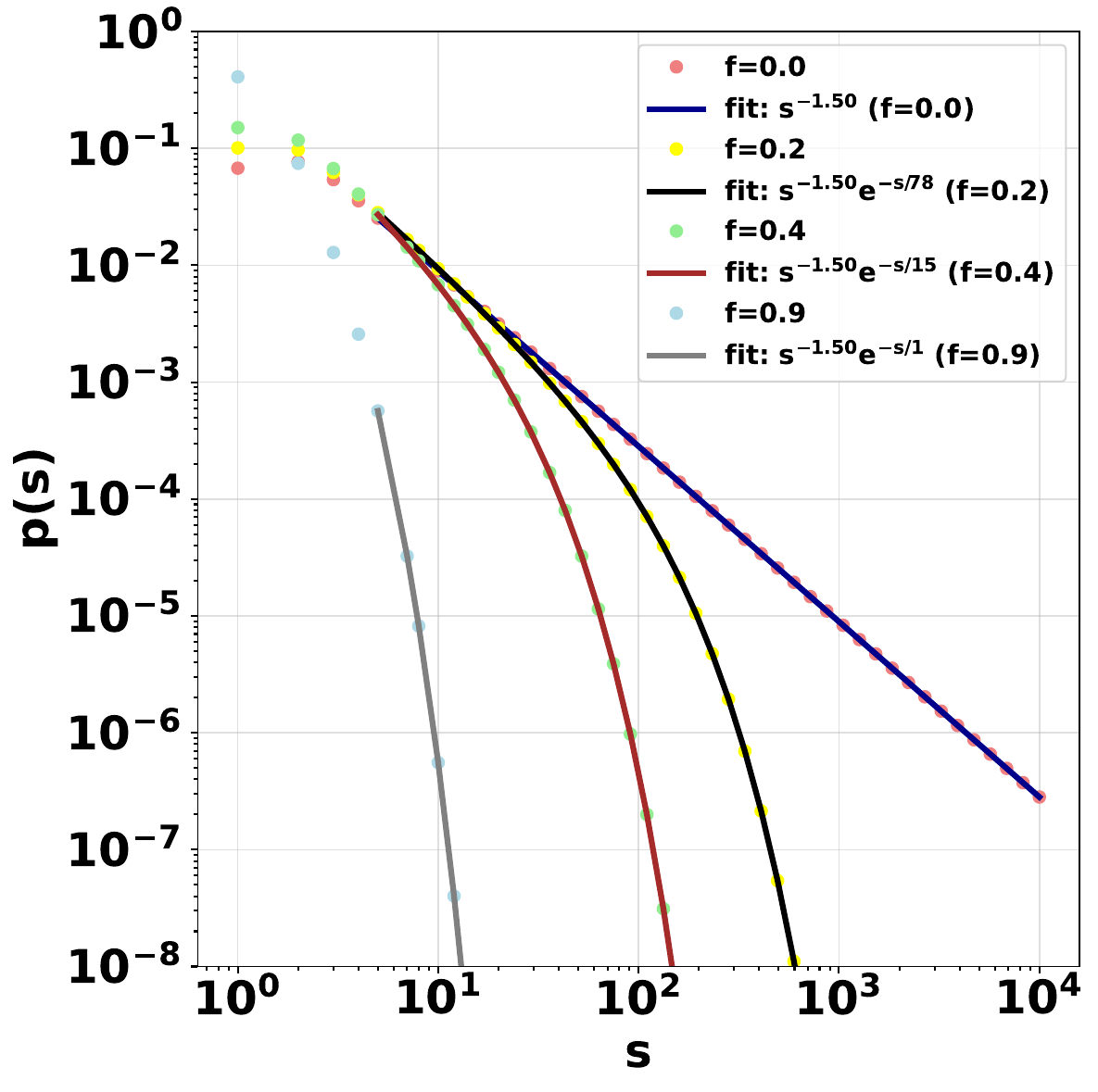}
    \caption{Avalanche size distributions for different dissipation probability $f$ in the Erdős–Rényi case with $\lambda=2$ as obtained from Theorem~\ref{thm:lagrange}. The  fitting range is $[5,10\,000]$.}
    \label{Br_er}
\end{figure}
\begin{figure}
    \centering
    \includegraphics[width=0.95\linewidth]{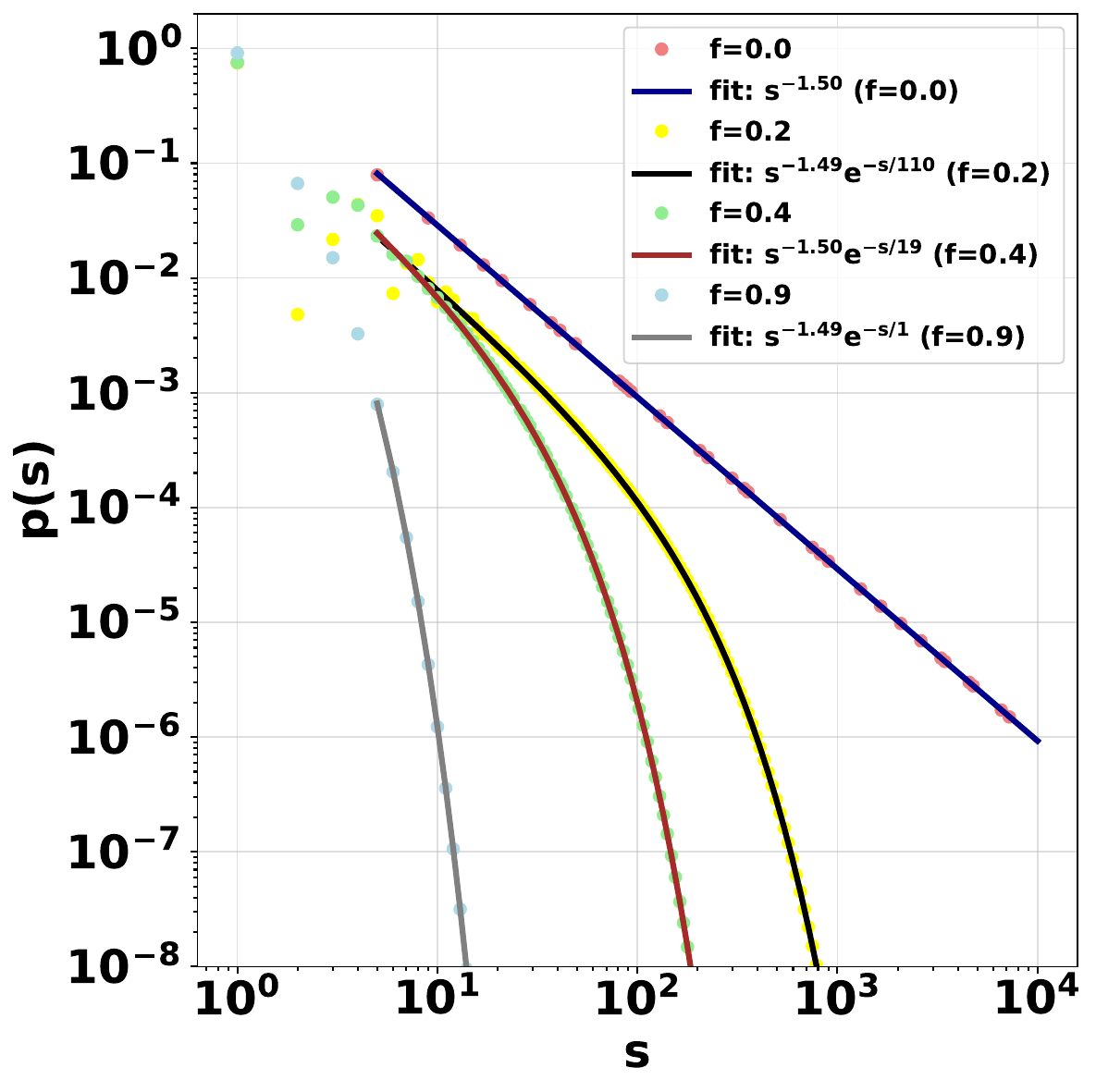}
    \caption{Avalanche size distributions for different dissipation probability $f$ in the random regular case with $k_0=4$ as obtained from Theorem~\ref{thm:lagrange}. The  fitting range is $[500,10\,000]$. }
    \label{Br_rr}
\end{figure}

\section{Impact of cycles}
\label{sec:cycles}
We now investigate how short cycles (in particular, triangles) affect the avalanche statistics of the BTW sandpile dynamics on scale-free networks. To this end, we consider the scale-free network model introduced by Holme and Kim~\cite{HolmeKim2002PRE}. This model generates networks with degree distribution $p_d(k)\sim k^{-\gamma}$, and it additionally allows us to tune the amount of clustering through a single parameter $p_{HK}\in [0,1]$. In the following we fixed $\gamma=3$ as in the Barab\'asi-Albert model~\cite{Babaalabert1999}.
The Holme-Kim construction is based on preferential attachment, but after attaching a new node to an existing node, one can close a triangle with probability $p_{HK}$ (triad formation). As a result, $p_{HK}$ directly controls the clustering coefficient:
\begin{itemize}
    \item For $p_{HK}=0$, no triangles are created. The network is essentially tree-like (locally close to a tree), and the clustering coefficient is approximately zero for very large graphs.
    \item For $p_{HK}=1$, the triad-formation step is performed whenever possible, producing a network with many triangles and a significantly larger clustering coefficient.
\end{itemize}
\begin{figure}
    \centering
    \includegraphics[width=1\linewidth]{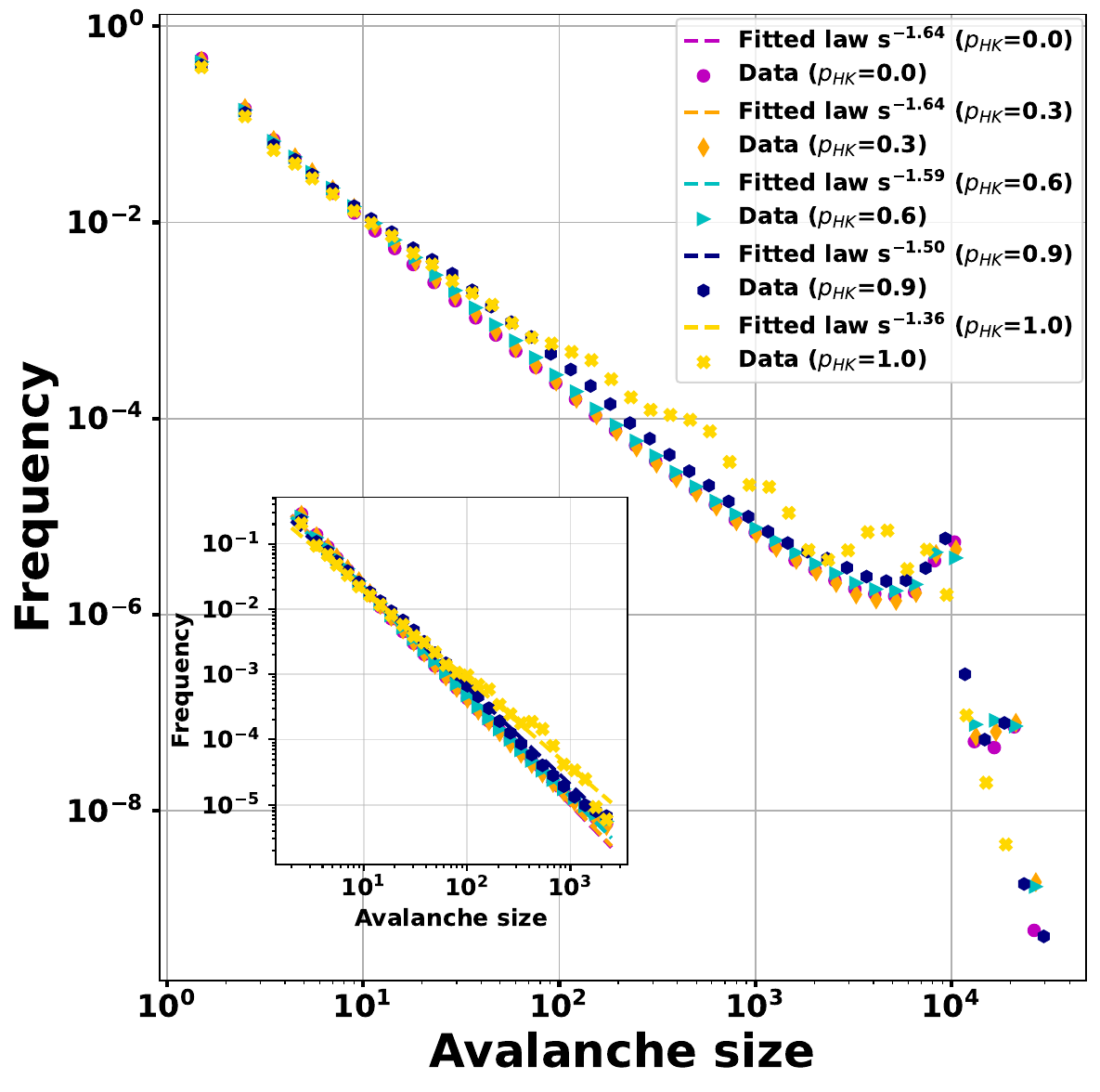}
   \caption{Main panel shows the distributions of avalanche sizes obtained with the sandpile model on different Holme-Kim scale-free graphs with $\gamma=3$. The number of nodes in each graph is $10\,000$. The dissipation probability is $5\times 10^{-3}$ in each case. The data are grouped into logarithmically spaced bins before being displayed on a log–log plot.  Inset shows distributions fitted on range $[1,2500]$. The goodness-of-fit diagnostics satisfy
$6.19\times10^{-3}\leq KS\leq1.22\times10^{-2}$,
$1.01\times10^{-2}\leq V\leq1.93\times10^{-2}$,
$5.76\times10^{-6}\leq A^2\leq4.14\times10^{-5}$,
and
$1.99\times10^{-3}\leq |\kappa-1|\leq3.21\times10^{-3}$.}
   \label{sp_Tsf}
\end{figure}
\begin{figure}
    \centering
    \includegraphics[width=1\linewidth]{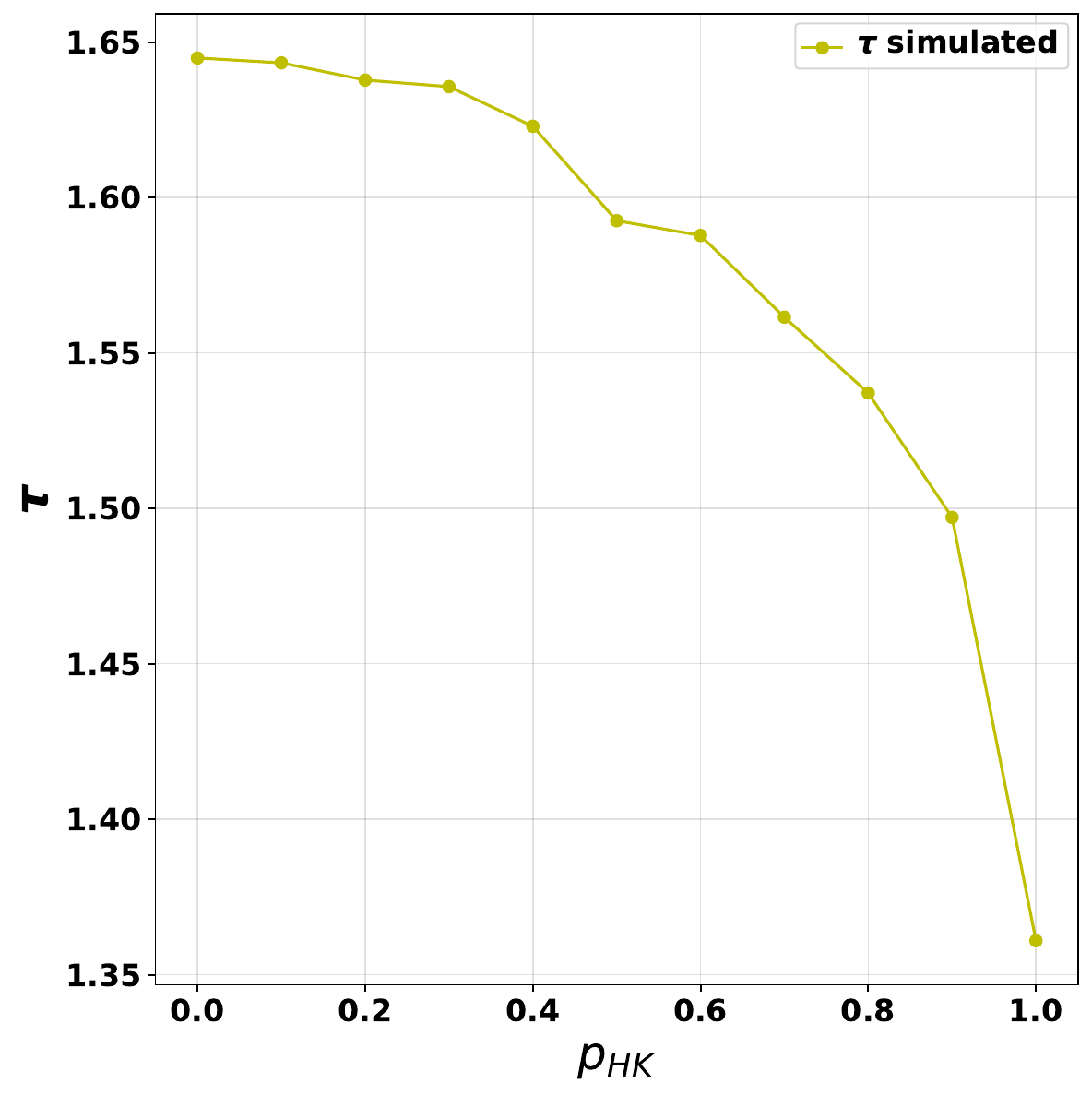}
   \caption{Variation of exponents $\tau$ we obtained by fitting the distribution by power-law with respect to $p_{HK}$ obtained from data shown in Figure~\ref{sp_Tsf}. The dissipation probability is $f=5\times 10^{-3}$.}
   \label{sp_Tsf_2}
\end{figure}

Importantly, while $p$ changes the abundance of small loops, the degree exponent remains close to $\gamma=3$, so we can isolate the effect of clustering without changing the overall scale-free nature of the topology. As already shown above, one can study avalanche dynamics in networks by approximating the spreading of activity as a branching process, for this to hold true one should typically assume that the underlying graph is a tree or at least locally tree-like. This assumption allows to consider different ``branches'' of the avalanche to evolve almost independently, which greatly simplifies the analysis. However, when the network contains many loops (especially triangles), this independence is no longer guaranteed: activity can return to already visited regions, different branches can merge, and correlations become stronger. Therefore, increasing $p$ provides a controlled way to test how the breakdown of the locally tree-like assumption modifies the avalanche-size distribution. We simulate the BTW sandpile model on Holme-Kim networks with $\gamma=3$ for several values of $p\in[0,1]$ and measure the avalanche size $s$ (the total number of topplings in an avalanche). The resulting distributions are shown in the main panel of Fig.~\ref{sp_Tsf}, where we plot the frequency of events as a function of $s$ on log-log axes, together with power-law fits in the scaling regime. Overall, the distributions remain close to a power law,
\begin{equation}
p(s)\sim s^{-\tau},
\end{equation}
but the exponent $\tau$ is not constant: it changes systematically with the clustering level.
For $\gamma=3$, the branching-process prediction above presented when there is no dissipation, gives $\tau=3/2$, while the previous result, for the case $\gamma=3$, returned an exponent $\tau\simeq 1.62$ (Fig.~\ref{sp_Sf}), which is already close to $3/2$. Consistently, for $p=0$ (the least clustered case), Fig.~\ref{sp_Tsf} yields $\tau\simeq 1.64$, in good agreement with the earlier estimate. As $p$ increases, triangles become more abundant and the fitted exponent decreases. This trend is summarized in the Fig.~\ref{sp_Tsf_2}, which reports the fitted $\tau$ as a function of $p$. We observe a clear monotonic drop from \(\tau \approx 1.64 \quad (p=0)\) down to \(
\tau \approx 1.36 \quad (p=1).\) 
In other words, stronger clustering makes the avalanche-size distribution heavier-tailed: large avalanches become relatively more frequent compared to the weakly clustered case. This behavior is consistent with the idea that clustering breaks the locally tree-like structure required by the branching-process approximation. When many triangles are present, the avalanche propagation is no longer well described by independent branching: feedback through short loops and the merging of propagation paths introduce correlations that effectively alter the scaling. The main qualitative conclusion is that small cycles are not a negligible detail: even at fixed $\gamma=3$, tuning the triangle-closing probability changes the measured critical exponent and thus the universality class suggested by the simplest tree-based theory. We now show that similar phenomena arise also on regular network topologies.

We thus apply the sandpile model to three regular network topologies.
All networks have the same number of nodes $N$ and the same constant degree $k_0=4$,
so differences in the results come only from the connectivity pattern. We have considered:
\begin{itemize}
\item {\em Random regular graph}: a regular graph where every node has degree $k_0$,
but edges are paired uniformly at random subject to this constraint. The network is homogeneous, yet it has no underlying spatial structure.
\item {\em Torus grid} $(100,100)$: a $100\times100$ square lattice with periodic boundary
conditions in both directions (equivalently, the lattice is wrapped into a torus). Each node connects to its four nearest neighbors (up, down, left, right), hence $k_0=4$.
\item {\em Regular ring}: nodes are placed on a cycle and connected locally. With $k_0=4$, each node links to its two nearest neighbors on the left and two on the right,
so all nodes have the same degree while preserving a one-dimensional geometry.
\end{itemize}
\begin{figure}[ht]
    \centering
    \includegraphics[width=1\linewidth]{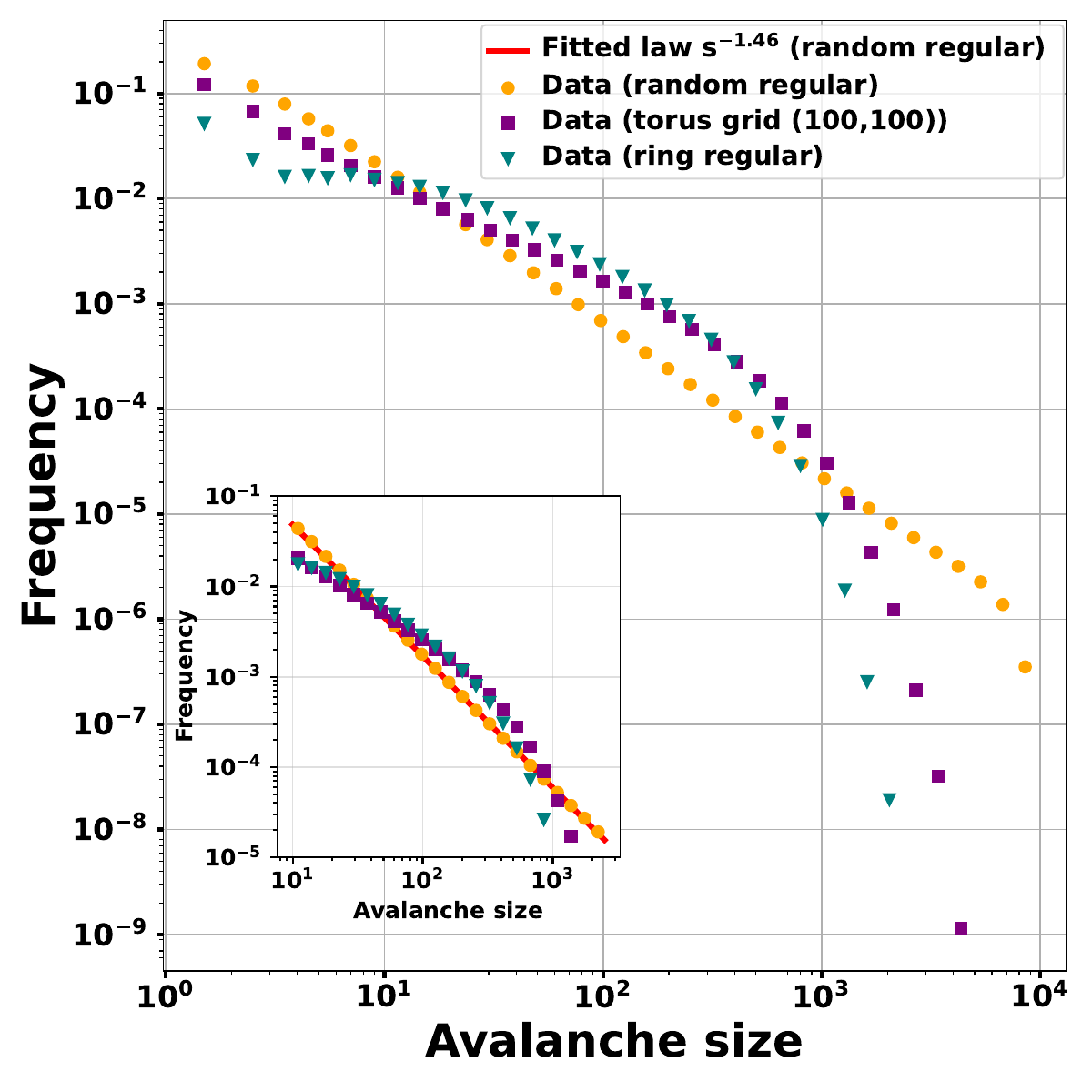}
   \caption{Distributions of avalanche sizes obtained with the sandpile model on different random regular graphs. The dissipation probability is $5\times 10^{-3}$ in each case. The data are grouped into logarithmically spaced bins before being displayed on a log–log plot. The fitting range is $[10,2500]$. The goodness-of-fit diagnostics satisfy
$KS=3.86\times10^{-3}$,
$V=6.41\times10^{-3}$,
$A^2=1.01\times10^{-4}$,
and
$|\kappa-1|=8.56\times10^{-5}$. }
   \label{Sp_Rn}
\end{figure}

The avalanche-size distributions are reported in Fig.~\ref{Sp_Rn}. A clear power-law regime is observed for the random regular network, which is consistent with the fact that random regular graphs are locally tree-like: in the large-$N$ limit, short cycles are rare and the neighborhood of a typical node resembles a tree over several steps. In contrast, no convincing power-law behavior appears for the two other graphs, because they contain many short loops and therefore deviate strongly from a locally tree-like geometry. For the torus grid, each node belongs to several 4-cycles, and the total number of elementary squares scales as the number of nodes, making short cycles ubiquitous. The regular ring 
is a purely ordered one-dimensional lattice with strong local clustering and a high density of short loops. Overall, these results suggest that local tree-likeness favors scale-free avalanche statistics, whereas abundant short cycles tend to suppress an extended power-law regime.
For the random regular network, we obtain a fitted exponent of $1.46$ for $k_0=4$, which is quasi equal to the value $1.5$ predicted by the branching-process approach.  
This small difference arises because the branching-process calculation relies on simplifying several expressions to capture the leading behavior of the distribution.

\subsection{Special case of trees}
\begin{figure}[ht]
    \centering
    \includegraphics[width=1\linewidth]{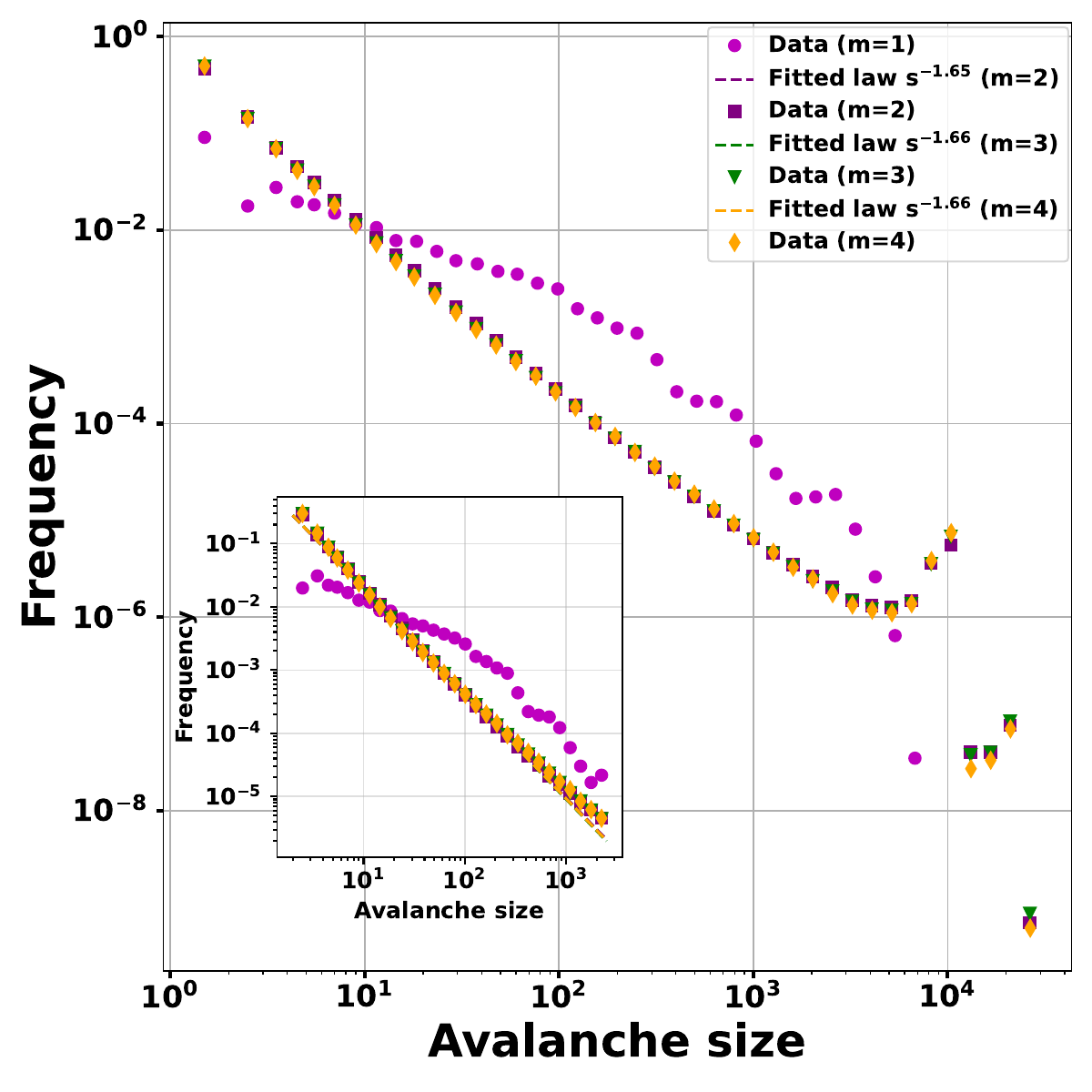}
   \caption{Distributions of avalanche sizes obtained with the sandpile model on different scale-free graphs obtained with the Barab\'asi-Albert model. The number of nodes in each graph is $10\,000$. The dissipation probability is $5\times 10^{-3}$ in each case. The data are grouped into logarithmically spaced bins before being displayed on a log–log plot. Inset shows distributions fitted on range $[1,2500]$. The goodness-of-fit diagnostics remain small across values of $m$, with
$KS\in[1.09,\,2.24]\times10^{-2}$,
$V\in[1.69,\,3.04]\times10^{-2}$,
$A^2\in[0.93,\,3.74]\times10^{-5}$,
while the normalized Kappa coefficient remains close to unity,
$|\kappa-1|\in[2.80,\,3.90]\times10^{-3}$.}
   \label{sp_bb}
\end{figure}
We now investigate the sandpile model with dissipation on scale-free networks generated by the Barab\'asi-Albert (BA) preferential-attachment model \cite{Babaalabert2002,Babaalabert1999}. We consider several values of the parameter $m$, which denotes the number of edges that each newly added node creates upon entering the network. For each network, we simulate the sandpile dynamics and measure the resulting avalanche-size distribution, shown in Fig.~\ref{sp_bb}. In the standard BA model, the degree distribution follows a power law with exponent $\gamma=3$. For $m\geq 2$, the avalanche-size distributions are well described by power-law, with exponents close to those obtained for the static model at $\gamma=3$. In contrast, the case $m=1$ yields a tree (a loopless network), however in this tree regime, the avalanche-size distribution does not exhibit a clear power-law scaling, suggesting that the sandpile dynamics on trees
does not produce avalanches whose sizes follow a power-law in the same way as in loopy networks. This result seems to be at odds with the fact that in a tree, the statistical independence of the avalanche branches, is satisfied. Observe also that for $m=1$ the average degree is also small. 

A similar conclusion holds true in the case of Cayley graphs, that also have a trees structure. Let us denote by Cayley graph $(n_1,n_2)$, the Cayley graph where each node has $n_1$ neighbors except the leaves and the number of generations is $n_2 $. 
 The avalanche size distribution on Cayley graphs has been derived analytically in \cite{DharMajumdar1990} and they obtain $p(s) \sim s^{-3/2}$ for large avalanche sizes, without taking into account dissipation but allowing grains to leave the system once in the leaves. We apply the sandpile model on the following Cayley graphs:
\begin{itemize}
  \item Cayley graph $(2,5000)$, we obtain a network of 10001 nodes;
  \item Cayley graph $(3,12)$, we obtain a network of 12286 nodes;
  \item Cayley graph $(4,8)$, we obtain a network of 13121 nodes;
  \item Cayley graph $(5,6)$, we obtain a network of 6826 nodes.
\end{itemize}
\begin{figure}
    \centering
    \includegraphics[width=1\linewidth]{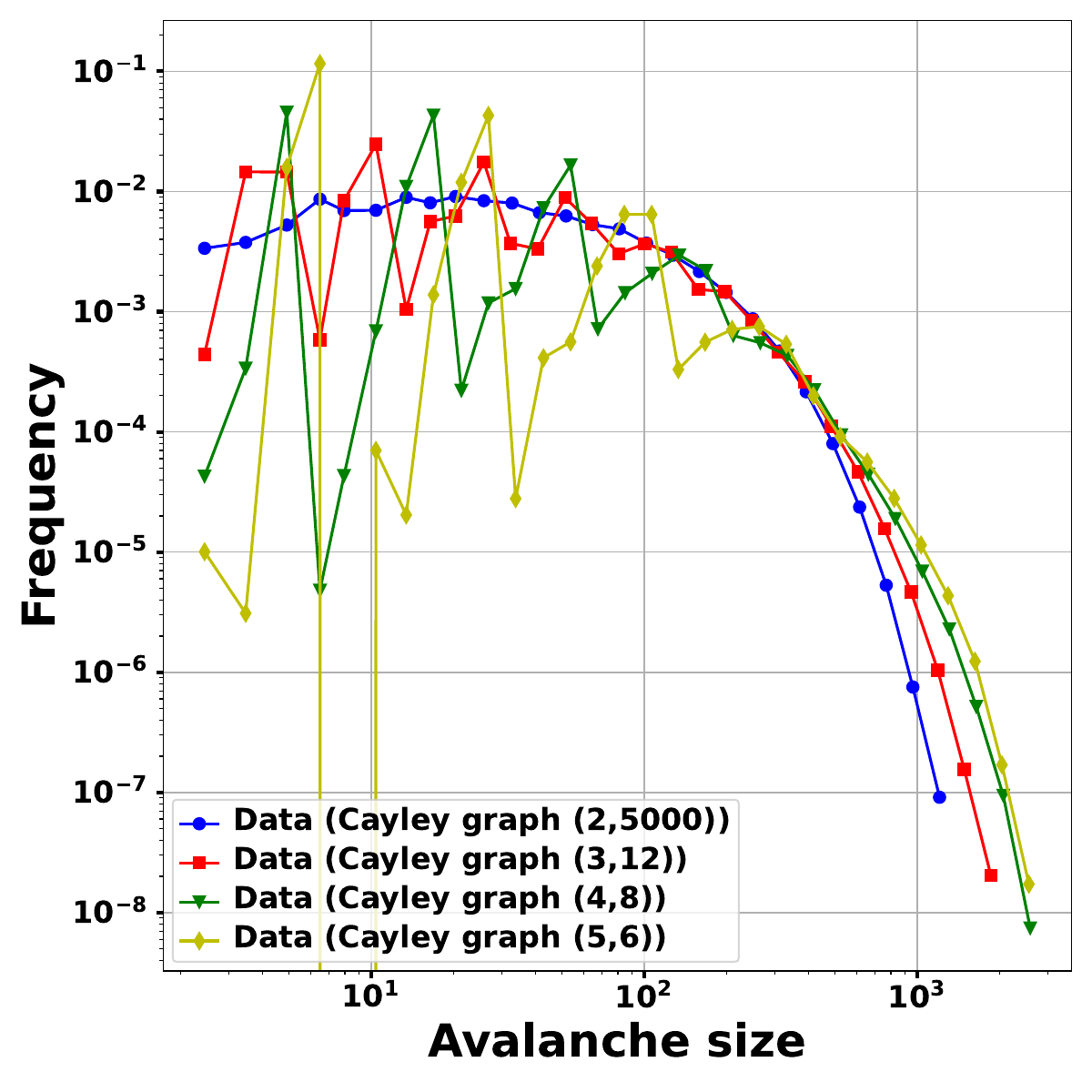}
   \caption{Distributions of avalanche sizes obtained with the sandpile model on different Cayley graphs. The dissipation probability is $5\times 10^{-3}$ in each case. We run the avalanche simulations $10$ times to obtain $10$ different samples of avalanche sizes. We merge the $10$ samples into a single sample. The data are grouped into logarithmically spaced bins before being displayed on a log–log plot.}
   \label{Sp_Ct}
\end{figure}
The avalanche-size distributions obtained on the different Cayley graphs are shown in Fig.~\ref{Sp_Ct}. We observe that these distributions do not follow a power law, even though the branching-process method predicts a power law for such graphs. Cayley graphs are trees, so one might expect the branching-process approximation to work well in this setting. However, in our case the agreement is poor because the network contains too few edges for the number of nodes in it. More precisely, a tree with $N$ nodes has exactly $N-1$ edges, which strongly limits the possible propagation paths of an avalanche. We already observed that with  Barab\'asi-Albert (BA)  model ($m=1$). When the number of edges is small, the branching-process approximation does not provide a good description of the avalanche statistics. A similar limitation was already observed in the previous section for random regular networks at small $k_0$. In that case, increasing $k_0$ increases the number of edges and the fitted exponent moves closer to the branching-process prediction. The explanation is that when the number of edges is small compared to the number of nodes, grains tend to quickly return to nodes they have already visited. This creates correlations between branches of the avalanche and violates the independence assumption underlying the branching-process approximation. In addition, the quality of the Cayley graph distributions worsens when we increase the number of neighbors per node. This phenomenon is explained by the presence of many leaves in these trees. Leaves have degree one, so when they receive a grain they can transfer it only to their unique neighbor. This favors rapid termination of avalanches and distorts the resulting size distribution. Therefore, despite the tree-like nature of Cayley graphs, the branching-process method does not reproduce the observed distributions in Fig.~\ref{Sp_Ct}.

\section{Conclusions}
\label{sec:conclu}
In this work, we investigated the sandpile model on a wide range of complex network topologies by combining analytical branching-process methods with extensive numerical simulations. A central contribution of the study was the extension of the classical branching-process framework to include dissipation during avalanche propagation. By introducing a probability of grain loss into the offspring distribution, we derived generalized generating functions capable of describing avalanche dynamics in dissipative environments. The analytical treatment showed that dissipation drives the system from the critical regime toward a subcritical one, leading to avalanche-size distributions characterized by power-law behavior with exponential cutoffs rather than purely scale-free statistics. Moreover we proved that the power law exponent moves from the value $\tau=\gamma/(\gamma-1)$, in the case $f=0$, to $\tau=\gamma$ for $0<f<1$.

A major objective of the work was to understand the impact of cycles and clustering on avalanche statistics. By using Holme--Kim clustered scale-free networks, we demonstrated that by increasing the density of triangles systematically modifies the avalanche-size distribution. In particular, the avalanche exponent decreases as clustering increases, implying heavier-tailed distributions and a larger probability of extreme cascades. These results clearly show that short cycles introduce strong correlations between avalanche branches and therefore invalidate the independence assumptions underlying classical branching-process theories. Consequently, clustering does not merely perturb the dynamics quantitatively, but can modify the avalanche-size distribution and thus provides evidence for a change in the underlying critical behavior.

The comparison between different regular network structures further emphasized the importance of local geometry. Random regular graphs displayed approximate power-law behavior because their local neighborhoods remain close to trees despite degree homogeneity. In contrast, regular lattices and ring-like structures exhibited strong deviations from power-law scaling due to the abundance of short loops and spatial correlations. These findings indicate that local tree-likeness is more important for branching-process behavior than degree regularity alone.
The analysis of trees produced particularly interesting results. Although branching-process methods are naturally motivated by tree structures, the simulations revealed that sparse trees do not necessarily exhibit the expected scale-free avalanche behavior. In Barabási--Albert trees with $m=1$ and in Cayley graphs, avalanche-size distributions deviated significantly from theoretical predictions. The results suggest that the low edge density and the large number of leaves strongly constrain avalanche propagation and introduce correlations that are not captured by ideal branching approximations. This demonstrates that the absence of loops alone is insufficient to guarantee branching-process behavior; connectivity density and structural heterogeneity also play essential roles. 
For very low values of the dissipation probability $f$, a noticeable bump appears in the tails of the distribution, as observed in Figs.~\ref{sp_bb}, \ref{sp_Tsf}, \ref{sp_Sf}, and \ref{sp_Er}. However, statistical analysis does not support the presence of dragon kings \cite{mikaberidze2023dragonkingsselforganizedcriticality}. This excess of large avalanches is more likely caused by finite-size and near-critical effects: when dissipation is weak, avalanches can spread through a large fraction of the network before dying out, by increasing correlations and enhancing the impact of the network topology on the avalanche dynamics.

Overall, this study highlights the strong sensitivity of the sandpile model to the underlying network structure. In particular, the distributions obtained on locally tree-like networks differ markedly from those observed on trees. Dissipation, clustering, cycles, and sparse connectivity all affect avalanche propagation in nontrivial ways. Although branching-process approaches provide valuable analytical insight for random and locally tree-like networks, their limitations become apparent in highly structured or strongly correlated systems.

\section*{Acknowledgments}
Computational resources have been provided by the Consortium des Équipements de Calcul Intensif (CÉCI), funded by the Fonds de la Recherche Scientifique de Belgique (F.R.S.-FNRS) under Grant No. 2.5020.11 and by the Walloon Region. Research is funded by ARC project "EMOTIONS".

\bibliography{article}

\newpage
\appendix

\section{Determination of the functional equation $P(y) = y\, Q\bigl(P(y)\bigr)$}\label{app1}

The aim of this Appendix is to prove the main functional equation~\eqref{SEC1EQ4} relating the probability generating function of the avalanche distribution size, $P(y)$, and the offspring distribution $\{q_k\}_{k \ge 0}$.

\begin{proof}
We decompose the random tree according to the first generation. The root of the avalanche tree (the initially toppled node) always exists and contributes one individual. Let $K$ be the number of children of the root, i.e., the number of nodes that topple as a direct consequence of the root toppling. Conditional on $K = k$, the root has $k$ children, giving rise to $k$ subtrees, one rooted at each child. Let $S_i$ denote the total size of the subtree rooted at the $i$-th child ($i = 1,\dots,k$). By the branching property, the random variables $S_1,\dots,S_k$ are independent and identically distributed, each with the same distribution as $S$. Thus, conditional on $K=k$,
\[
  S = 1 + \sum_{i=1}^{k} S_i.
\]
For any $s \ge 1$, conditioning on $K$ yields
\begin{equation}\label{APP1EQ1}
p(s) = \mathbb{P}(S = s)
  = \sum_{k=0}^{\infty} q_k\, \mathbb{P}(S = s \mid K = k).
\end{equation}
Fix $k \ge 0$.
The event $\{S = s\}$ is equivalent to
\[
  1 + \sum_{i=1}^{k} S_i = s
  \quad \Longleftrightarrow \quad
  \sum_{i=1}^{k} S_i = s-1.
\]
Let $\alpha_1,\dots,\alpha_k \ge 1$ be integers. Since the subtrees are independent and each $S_i$ has distribution $\{p(\cdot)\}$, we have
\[
 \mathbb{P}(S_1=\alpha_1,\dots,S_k=\alpha_k)
  = \prod_{i=1}^{k} p(\alpha_i).
\]
Therefore, conditional on $K=k$,
\begin{align*}
\mathbb{P}(S = s \mid K = k)
  &= \sum_{\alpha_1,\dots,\alpha_k \ge 1}
     \left[\prod_{i=1}^{k} p(\alpha_i)\right]
     \delta_{\alpha_1 + \cdots + \alpha_k,\,s-1},
\end{align*}
where $\delta_{a,b}$ is the Kronecker symbol. Plugging this expression into
\eqref{APP1EQ1} gives
\begin{equation}\label{APP1EQ2}
p(s)
  = \sum_{k=0}^{\infty} q_k
    \sum_{\alpha_1,\dots,\alpha_k \ge 1}
    \left[\prod_{i=1}^{k} p(\alpha_i)\right]
    \delta_{\alpha_1 + \cdots + \alpha_k,\,s-1}.
\end{equation}
Multiply \eqref{APP1EQ2} by $y^s$ and sum over $s \ge 1$
\begin{align*}
P(y)
  &= \sum_{s=1}^{\infty} p(s)\, y^s \\
  &= \sum_{s=1}^{\infty}
     \sum_{k=0}^{\infty} q_k
     \sum_{\alpha_1,\dots,\alpha_k \ge 1}
     \left[\prod_{i=1}^{k} p(\alpha_i)\right]
     \delta_{\alpha_1 + \cdots + \alpha_k,\,s-1}\, y^s.
\end{align*}
Since all terms are nonnegative, we may interchange the order of summation
\begin{align*}
P(y)
 &= \sum_{k=0}^{\infty} q_k
    \sum_{\alpha_1,\dots,\alpha_k \ge 1}
    \left[\prod_{i=1}^{k} p(\alpha_i)\right]
    \sum_{s=1}^{\infty}
    \delta_{\alpha_1 + \cdots + \alpha_k,\,s-1}\, y^s.
\end{align*}
For fixed $\alpha_1,\dots,\alpha_k$, the Kronecker symbol enforces
$s-1 = \alpha_1 + \cdots + \alpha_k$, so
\[
\sum_{s=1}^{\infty}
   \delta_{\alpha_1 + \cdots + \alpha_k,\,s-1}\, y^s
 = y^{1 + \alpha_1 + \cdots + \alpha_k}.
\]
Thus
\begin{align*}
P(y)
 &= \sum_{k=0}^{\infty} q_k
    \sum_{\alpha_1,\dots,\alpha_k \ge 1}
    \left[\prod_{i=1}^{k} p(\alpha_i)\right]
    y^{1 + \alpha_1 + \cdots + \alpha_k} \\
 &= y \sum_{k=0}^{\infty} q_k
    \sum_{\alpha_1,\dots,\alpha_k \ge 1}
    \prod_{i=1}^{k} \bigl(p(\alpha_i)\,y^{\alpha_i}\bigr).
\end{align*}
The inner sums factorize,
\begin{align*}
P(y)
 &= y \sum_{k=0}^{\infty} q_k
    \prod_{i=1}^{k}
    \left(\sum_{\alpha_i=1}^{\infty} p(\alpha_i)\,y^{\alpha_i}\right) \\
 &= y \sum_{k=0}^{\infty} q_k \bigl(P(y)\bigr)^k.
\end{align*}
The last equality follows from the definition of $P(y)$. Recognizing the generating function $Q(w)$, we obtain
\[
   P(y)
   = y\, Q\bigl(P(y)\bigr),
\]
which is exactly \eqref{SEC1EQ4}.
\end{proof}

\section{Computation of generating functions and means in presence of dissipation}\label{app2}

In this Section we will provide detailed computations to obtain the probability generating function of the offspring distribution, $Q(w|f)$, and the mean, $\mu$. 

\paragraph{Scale-Free case}

Scale-free networks are characterized by a degree distribution given by
\begin{equation*}
   p_d(k) \sim k^{-\gamma}\, ,
\end{equation*}
with $2 < \gamma < 3$ in many empirical networks~\cite{Babaalabert1999,newmanbook}. By assuming the minimum degree to be $k_{\min}=1$, we compute
\begin{equation*}
 q_k = \frac{k^{-\gamma}}{\zeta(\gamma-1)},
 \qquad k \ge 1\, ,
\end{equation*}
and thus
\begin{equation*}
     q_0
     = 1 - \sum_{k=1}^{\infty} \frac{k^{-\gamma}}{\zeta(\gamma-1)} \notag = 1 - \frac{\zeta(\gamma)}{\zeta(\gamma-1)}\, ,
\end{equation*}
where $\zeta(\gamma) = \sum_{n=1}^{\infty} n^{-\gamma}$ is the Riemann zeta function.

From the definition~\eqref{eq:Qwffun} we can compute
\begin{align*}
Q(w|f) 
&= \sum_{k \ge 0} q_k (f)w^k \\[6pt]
&= 1 - \frac{\zeta(\gamma)}{\zeta(\gamma-1)}
 \\&  + \frac{1}{\zeta(\gamma-1)}
     \sum_{k=0}^{\infty} \sum_{m=k}^{\infty}
     \binom{m}{k} m^{-\gamma} (1-f)^k f^{m-k}w^k \\[6pt]
&= 1 - \frac{\zeta(\gamma)}{\zeta(\gamma-1)}
\\&   + \frac{1}{\zeta(\gamma-1)}
     \sum_{m=0}^{\infty} m^{-\gamma}
     \sum_{k=0}^{m}
     \binom{m}{k} (1-f)^k f^{m-k}w^k \\[6pt]
&= 1 - \frac{\zeta(\gamma)}{\zeta(\gamma-1)}
\\&   + \frac{1}{\zeta(\gamma-1)}
     \sum_{m=0}^{\infty} m^{-\gamma}
     \big[w(1-f)+f\big]^m \\[6pt]
&= 1 - \frac{\zeta(\gamma)}{\zeta(\gamma-1)}
 \\&  + \frac{1}{\zeta(\gamma-1)}
     \sum_{k=0}^{\infty}
     \frac{\mu(f)^k}{k^\gamma}\\
     Q(w|f)
&=
1 - \frac{\zeta(\gamma)}{\zeta(\gamma-1)}
+ \frac{\mathrm{Li}_{\gamma}\big(u(w)\big)}{\zeta(\gamma-1)}
\end{align*}

Similarly we can compute the average number of offsprings produced by a toppling event~\eqref{eq:muf}
\begin{align*}
\mu(f)
&= \sum_{k \ge 0} k q_k (f)\\
&= \sum_{k=0}^{\infty} k
   \sum_{m=k}^{\infty}
   \binom{m}{k} q_m (1-f)^k f^{m-k} \\[6pt]
&= \sum_{k=0}^{\infty} k
   \sum_{m=k}^{\infty}
   \binom{m}{k}
   \frac{m^{-\gamma}}{\zeta(\gamma-1)}
   (1-f)^k f^{m-k} \\[6pt]
&= \sum_{m=1}^{\infty}
   \frac{m^{-\gamma}}{\zeta(\gamma-1)}
   \sum_{k=1}^{m}
   k \binom{m}{k}
   (1-f)^k f^{m-k} \\[6pt]
&= \sum_{m=1}^{\infty}
   \frac{m^{-\gamma}}{\zeta(\gamma-1)}
   m(1-f+f)^{m-1}  \\[6pt]
&= \sum_{m=1}^{\infty}
   \frac{m^{-\gamma+1}}{\zeta(\gamma-1)}
   (1-f)\\
   \mu(f)&=1-f
\qquad \text{with } f>0  
\end{align*}

\paragraph{Erd\H{o}s--R\'enyi Case with $\lambda>1$}

The Erd\H{o}s-R\'enyi (ER) random graph exhibits the probability degree distribution 
\begin{equation*}
    p_d(k) = e^{-\lambda} \frac{\lambda^k}{k!}\, ,
\end{equation*}
where $\lambda = p(N-1)$, $p\in (0,1)$ is the probability that each pair of nodes is connected to form a link, and $N$ is the number of nodes, assumed to be large. We can thus obtain
\begin{equation*}
    q_k
     = \frac{\lambda^{k-1} e^{-\lambda}}{k!},
    \qquad k \ge 1\, ,
\end{equation*}
and
\begin{equation*}
 q_0
 = 1 - \frac{1}{\lambda} + \frac{e^{-\lambda}}{\lambda}\, ,
\end{equation*}
and eventually
\begin{align*}
Q(w\mid f )
&= \sum_{k \ge 0} q_k(f) w^k \\
&= 1 - \sum_{k=1}^{\infty} q_k
  \\& + \sum_{k=0}^{\infty} \sum_{m=k}^{\infty}
   \binom{m}{k}
   q_m
   (1-f)^k f^{m-k} \, w^k \\
   &= 1 - \frac{1}{\lambda}
\\&   + \frac{1}{\lambda}
       \sum_{k=0}^{\infty} \sum_{m=k}^{\infty}
   \binom{m}{k}
   \frac{\lambda^m e^{-\lambda}}{m!}
   (1-f)^k f^{m-k} \, w^k \\ 
&= 1 - \frac{1}{\lambda}
 \\&  + \frac{1}{\lambda}
   \sum_{m=0}^{\infty}
   \frac{\lambda^m e^{-\lambda}}{m!}
   \sum_{k=0}^{m}
   \binom{m}{k}
   (1-f)^k f^{m-k} w^k  \\
&= 1 - \frac{1}{\lambda}
 \\&  + \frac{1}{\lambda}
   \sum_{m=0}^{\infty}
   \frac{\lambda^m e^{-\lambda}}{m!}
   \big(f + (1-f)w\big)^m \\
   &= 1 - \frac{1}{\lambda}
 \\&  + \frac{1}{\lambda}
   \sum_{m=0}^{\infty}
   \frac{\lambda^m e^{-\lambda}}{m!}
   \big(u(w)\big)^m \\
   Q(w\mid f)&
= 1 - \frac{1}{\lambda}
   + \frac{e^{\lambda (u(w)-1)}}{\lambda}\\
\end{align*}

An analogous computation allows to get
\begin{align*}
\mu(f)
&= \sum_{k \ge 0} k\, q_k(f) \\
&= \frac{1}{\lambda}\sum_{k=0}^{\infty} \sum_{m=k}^{\infty}
k \binom{m}{k}
\frac{\lambda^m e^{-\lambda}}{m!}
(1-f)^k f^{m-k} \\
&= \frac{1}{\lambda}\sum_{k=0}^{\infty}k\frac{\lambda^k e^{-\lambda}}{k!}
(1-f)^k\sum_{m=k}^{\infty}
\frac{\lambda^{m-k} }{(m-k)!}
f^{m-k} \\
&= \frac{e^{\lambda f}}{\lambda}\sum_{k=0}^{\infty}k\frac{\lambda^k e^{-\lambda}}{k!}
(1-f)^k\\
&= \frac{e^{-\lambda(1- f)}}{\lambda}\sum_{k=0}^{\infty}k\frac{\lambda^k} {k!}
(1-f)^k\\
&= \frac{1}{\lambda}
\lambda(1-f)\\
\mu(f)&=1-f
\qquad \text{with } f>0  
\end{align*}

The root of Eq.~\eqref{eq:wstarER} involves the Lambert W-function and thus it cannot be computed explicitly. In Fig.~\eqref{Er} we report the dependence of the root $w^*$ on the dissipation $f$.
\begin{figure}[H]
    \centering
    \includegraphics[width=0.95\linewidth]{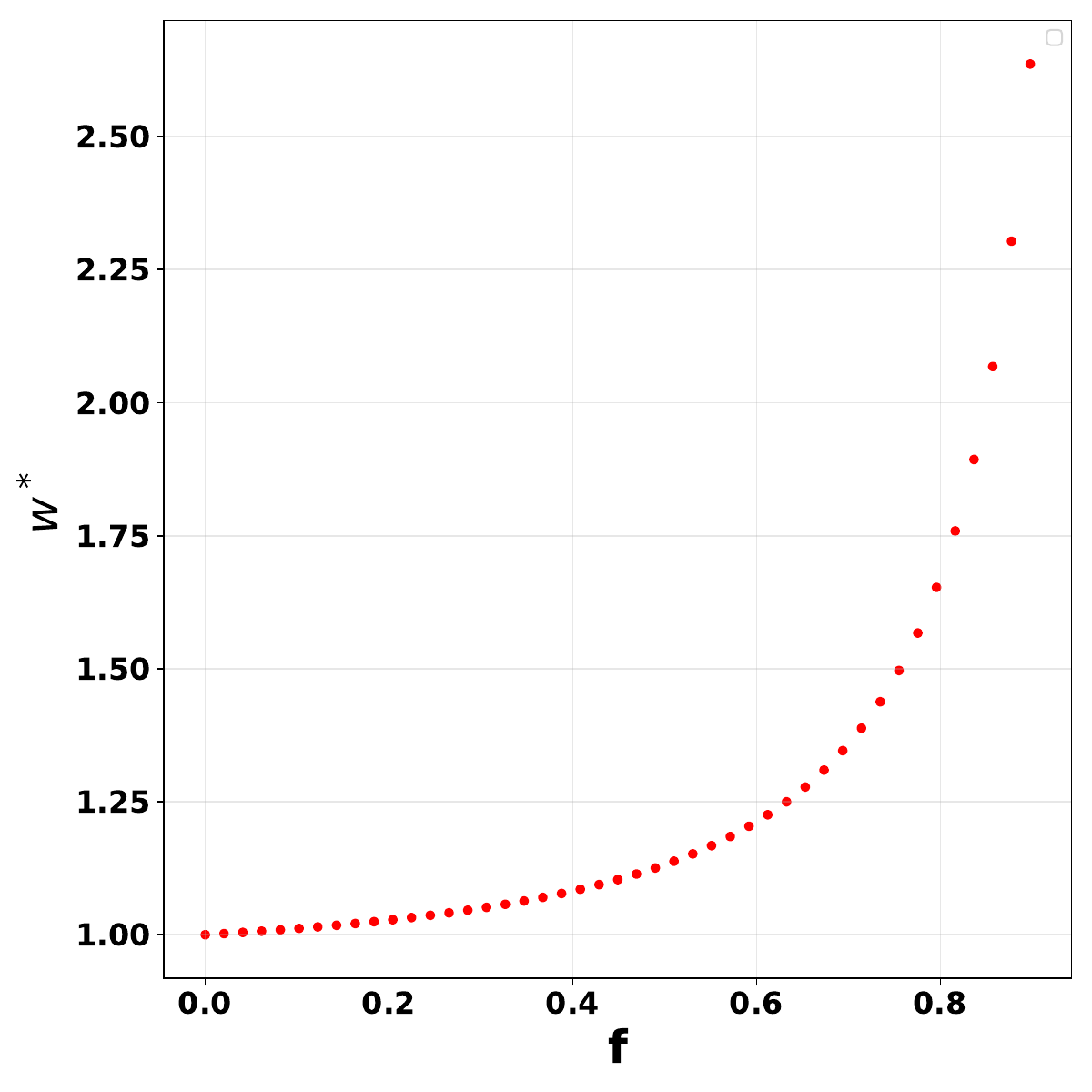}
    \caption{We show the root $w^*$ as a function of the dissipation $f$ in the Erd\H{o}s--R\'enyi case with $\lambda=10$.}
    \label{Er}
\end{figure}

In Fig.~\ref{sp_Er} we report the avalanche size probability distribution $p(s)$ for different values of the probability $p$ to have a link among two nodes.
\begin{figure}[H]
    \centering
    \includegraphics[width=0.95\linewidth]{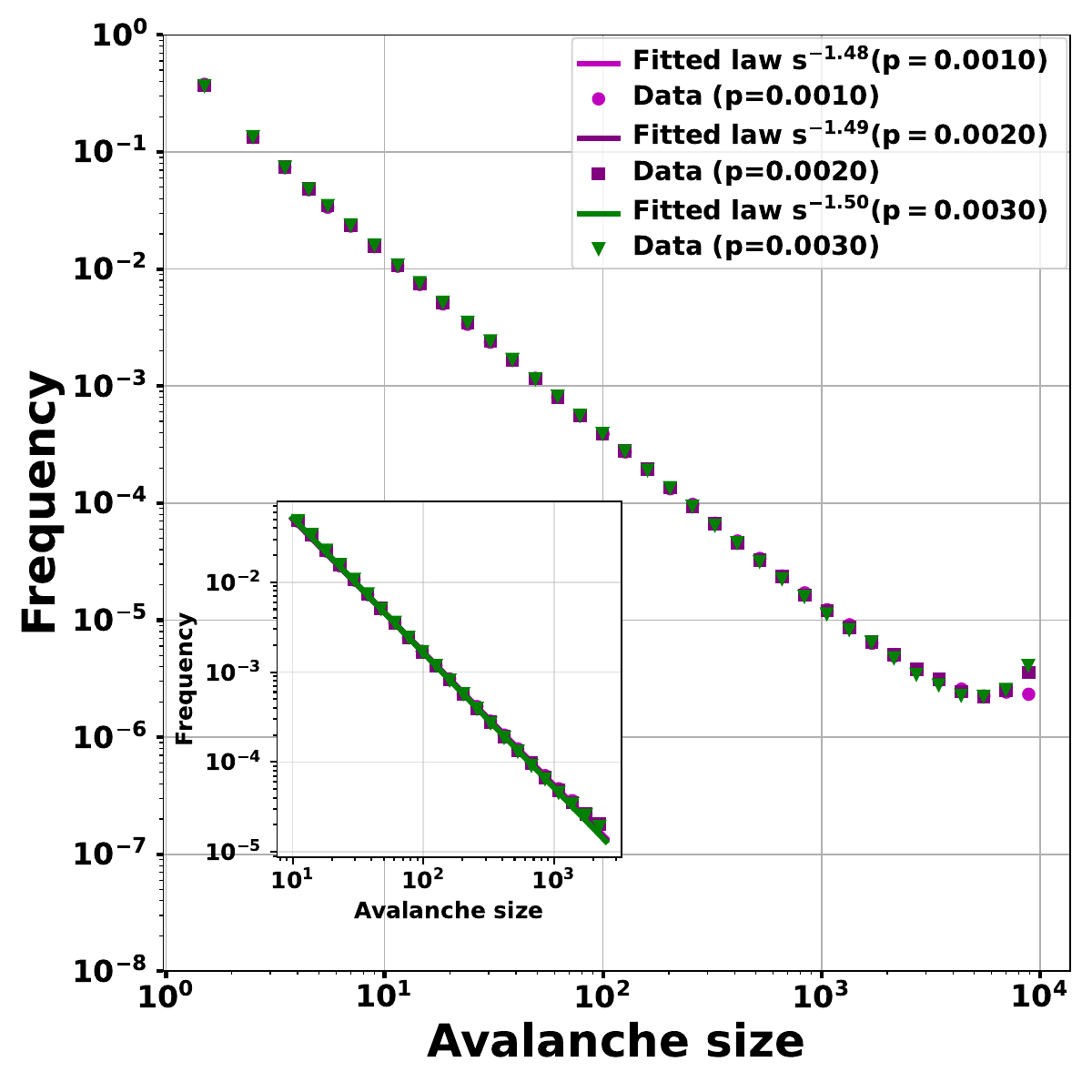}
   \caption{Distributions of avalanche sizes obtained with the sandpile model on different Erd\H{o}s-R\'enyi graphs. The dissipation probability is $5\times 10^{-3}$ in each case. The data are grouped into logarithmically spaced bins before being displayed on a log–log plot. The fitting range is $[10,2500].$ The goodness-of-fit diagnostics satisfy
$3.92\times10^{-3}\leq K_s\leq4.81\times10^{-3}$,
$7.63\times10^{-3}\leq V\leq9.23\times10^{-3}$,
$1.61\times10^{-3}\leq A^2\leq2.06\times10^{-3}$,
and
$0.00206\leq|\kappa-1| \leq0.00258$ .}
   \label{sp_Er}
\end{figure}
\paragraph{Random regular Case with $k_0>2$}

The degree distribution for a random $k_0$-regular graph~\cite{Mckaywormald1990} is simply given by 
\begin{equation*}
    p_d(k) = \delta_{k,k_0}\,,
\end{equation*}
hence
\begin{eqnarray*}
 q_{k_0} &=& \frac{1}{k_0}\\
 q_k &=& 0 \text{ if $k \ge 1$ $\neq k_0$}\\
   q_0 &=& \frac{k_0 - 1}{k_0}\, .
\end{eqnarray*}
Hence the associated generating function is
\begin{align*}
Q(w|f)=&\sum_{k=0}^{k_0}q_k(f) w^k\\
      =&1-q_{k_0}\\&+\sum_{k=0}^{k_0}q_{k_0}\binom{k_0}{k}(1-f)^k f^{k_0-k} w^k\\
      =&1-\frac{1}{k_0}+\frac{1}{k_0}\sum_{k=0}^{k_0} \binom{k_0}{k}(1-f)^k f^{k_0-k} w^k\\
      =&1-\frac{1}{k_0}+\frac{1}{k_0}\left((1-f)w +f\right)^{k_0}\\
 Q(w|f)  =&\frac{k_0-1}{k_0}+\frac{u(w)^{k_0}}{k_0} 
\end{align*}
with $u(w) = (1-f)w + f$. While the average offspring number is
\begin{align*}
\mu(f)&=\sum_{k \ge 0}k q_k(f)\\
&=\sum_{k \ge 0}
k \frac{1}{k_0}
\binom{k_0}{k}
(1-f)^k f^{k_0-k}\\
&=\frac{1}{k_0}
k_0 (1-f)\\
\mu(f)&=1-f
\qquad \text{with } f>0.
\end{align*}

It is not possible, in general, to analytically determine the root of Eq.~\eqref{eq:polyk0}
\begin{equation*}
(1-k_0)y^{k_0}+f y^{k_0-1}+k_0-1=0\, ,
\end{equation*}
one can however prove that there is a single positive root by using the Descartes rule of signs. In the following Fig.~\ref{Rr} we provide a numerical support to this claim in the case $k_0=10$.
  \begin{figure}[H]
    \centering
    \includegraphics[width=0.95\linewidth]{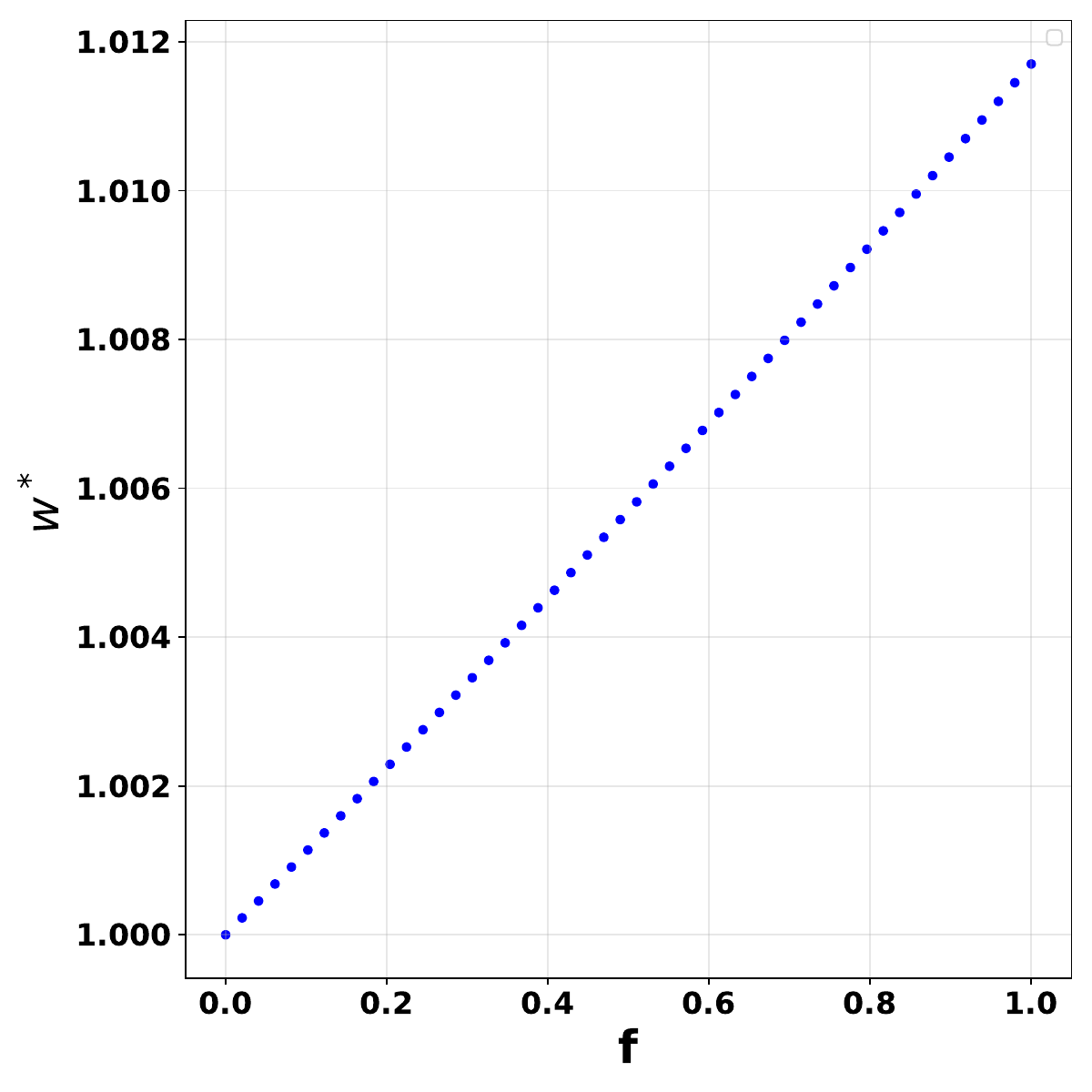}
    \caption{We show the root $w^*$ as a function of the dissipation probability $f$ in the random regular case with $k_0=10$.}
    \label{Rr}
\end{figure}

\section{Determination of  the asymptotic expansion of
$p(s)$ for large $s$.} \label{app4}

To analyze the singularity of $Q(w)$ at $w=1$, we use the Bose-Einstein type expansion \cite{robinson1951bose} valid for $s \in \mathbb{C} \setminus \mathbb{N}$
\begin{align*}
\mathrm{Li}_s(w)
&= \Gamma(1 - s)\, w^{s - 1}
+ \sum_{n=0}^{\infty} \frac{(-1)^n}{n!}\,\zeta(s - n)\, w^n\, .
\end{align*}
When $s$ is a positive integer, one has to handle the pole at $s=1$ separately. For $s \in \mathbb{N}$, the expansion can be written as
\begin{align}
\mathrm{Li}_s(w)
&= \lim_{m \to s} \Biggl\{
\Gamma(1 - m)\, w^{m - 1}
\\&+ \frac{(-1)^{s-1}}{(s-1)!}\, \zeta(m - s + 1)\, w^{s - 1}
\Biggr\}
\\&+ \sum_{n \neq s-1} \frac{(-1)^n}{n!}\, \zeta(s - n)\, w^n \notag \\
&= \lim_{m \to s} \Biggl\{
\Gamma(1 - m + s - s)\, w^{m - 1}
\\&+ \frac{(-1)^{s-1}}{(s-1)!}\, \zeta(m - s + 1)\, w^{s - 1}
\Biggr\}
\\&+ \sum_{n \neq s-1} \frac{(-1)^n}{n!}\, \zeta(s - n)\, w^n.
\label{app4eq1}
\end{align}
Set $\eta = m - s$ in \eqref{app4eq1}. Then
\begin{align*}
\mathrm{Li}_s(w)
&= \lim_{\eta \to 0} \Biggl\{
\Gamma(1 - s - \eta)\, w^{s - 1 + \eta}
\\&+ \frac{(-1)^{s-1}}{(s-1)!}\, \zeta(1+\eta)\, w^{s - 1}
\Biggr\}
\\&+ \sum_{n \neq s-1} \frac{(-1)^n}{n!}\, \zeta(s - n)\, w^n \\
&= \lim_{\eta \to 0} w^{s - 1} \Biggl\{
\Gamma(1 - s - \eta)\, w^{\eta}
\\&+ \frac{(-1)^{s-1}}{(s-1)!}\, \zeta(1+\eta)
\Biggr\}
\\&+ \sum_{n \neq s-1} \frac{(-1)^n}{n!}\, \zeta(s - n)\, w^n .
\end{align*}
We now expand the singular terms in $\eta$ as $\eta \to 0$. First,
\[
\zeta(1+\eta) = \frac{1}{\eta} + C + o(1),
\]
where $C$ is the Euler-Mascheroni constant. Next,
\[
w^{\eta} = e^{\eta \log w} = 1 + \eta \log w + o(\eta).
\]
For the gamma function we use
\[
\Gamma(1-s-\eta)
= \frac{(-1)^{s-1}}{(s-1)!}\left(-\frac{1}{\eta} - C + \sum_{n=1}^{s-1} \frac{1}{n} + o(1)\right).
\]
combining these expansions yields
\begin{equation}
\begin{aligned}
\Gamma(1-s-\eta)\, w^\eta
=
\frac{(-1)^{s-1}}{(s-1)!}
\Biggl(
&\frac{1}{\eta} - C + \sum_{n=1}^{s-1}\frac{1}{n} \\
&\quad - \log w + o(1)
\Biggr).
\end{aligned}
\end{equation}
After cancellation of the simple pole in $1/\eta$, we arrive at
\begin{align*}
\mathrm{Li}_s(w)
&= \frac{(-1)^{s-1}}{(s-1)!}\, w^{s - 1}
   \left(\sum_{n=1}^{s-1}\frac{1}{n} - \log w\right)
  \\& + \sum_{n \neq s-1} \frac{(-1)^n}{n!}\, \zeta(s - n)\, w^n \\
&= \frac{(-1)^{s-1}}{(s-1)!}\, w^{s - 1}\bigl(\psi(s) + C - \log w\bigr)
   \\&+ \sum_{n \neq s-1} \frac{(-1)^n}{n!}\, \zeta(s - n)\, w^n,
\end{align*}
where $\psi(s)$ is the digamma function. Specializing these expansions to $s = \gamma$ allows one to determine the singular behavior of $Q(w)$ near $w=1$. We eventually find.

\paragraph{Case $2 < \gamma < 3$:}
\begin{equation}
    Q(w)
    = \frac{\Gamma(1-\gamma)}{\zeta(\gamma-1)}(1-w)^{\gamma-1} + w
      + o\bigl((1-w)^{\gamma-1}\bigr).
\end{equation}

\paragraph{Case $\gamma = 3$.}
\begin{align}
    Q(w)
    &= 1 - \frac{\zeta(3)}{\zeta(2)}
       \nonumber+ \frac{1}{\zeta(2)}
         \Bigl[
           \bigl( 1+\tfrac{1}{2} - \log(1-w)\bigr)\tfrac{(1-w)^2}{2}
          \\& + \zeta(3) - \zeta(2)(1-w)
         \Bigr]
      \nonumber+ o\bigl((1-w)^2\bigr) \notag \\
    &= w + \frac{1}{\zeta(2)}
       \Bigl(\tfrac{3}{4} - \tfrac{1}{2}\log(1-w)\Bigr)(1-w)^2\notag \\
       &+ o\bigl((1-w)^2\bigr) \notag \\
    &= w - \frac{1}{2\zeta(2)}\log(1-w)\,(1-w)^2\notag\\
      & + o\bigl((1-w)^2 |\log(1-w)|\bigr).
\end{align}

\paragraph{Case $\gamma > 3$.}
In this case, the second moment of the degree distribution is finite and
\begin{align}
    Q(w) 
    &= 1 - \frac{\zeta(\gamma)}{\zeta(\gamma-1)}
       + \frac{1}{\zeta(\gamma-1)}\notag
         \Bigl[\zeta(\gamma) \\&- \zeta(\gamma-1)(1-w)
               + \frac{\zeta(\gamma-2)}{2}(1-w)^2\Bigr]\notag
     \\&  + o\bigl((1-w)^2\bigr) \notag\\
    &= w + \frac{\zeta(\gamma-2)}{2\,\zeta(\gamma-1)}(1-w)^2
       + o\bigl((1-w)^2\bigr).
\end{align}
To relate the singularity of $Q$ near $w=1$ to that of $P$ near $y=1$, we use \eqref{SEC1EQ5}. Write
\begin{align*}
   y= \frac{w}{Q(w)}= \frac{w}{w + (Q(w)-w)}. 
\end{align*}
Let us set $w = 1 - \varepsilon$ with $\varepsilon \downarrow 0$ and suppose that
\[
Q(w) = 1 - \varepsilon + b\,\varepsilon^{a} + o\bigl(\varepsilon^{a}\bigr),
\qquad a>1,
\]
for some constant $b \neq 0$. Then
\begin{align*}
y
  = \frac{1-\varepsilon}{1-\varepsilon + b\,\varepsilon^{a} + o(\varepsilon^a)}&\Longrightarrow y = 1 - b\,\varepsilon^{a} + o(\varepsilon^a), \\
 & \Longrightarrow  1-y
  \sim b\,\varepsilon^{a}, \\
&\Longrightarrow \varepsilon
  \sim \bigl(b^{-1}(1-y)\bigr)^{1/a}, \\
&\Longrightarrow w
  \sim 1 - \bigl(b^{-1}(1-y)\bigr)^{1/a}.
\end{align*}
Since $w = P(y)$, this yields
\[
P(y) \sim 1 - \bigl(b^{-1}(1-y)\bigr)^{1/a},
\qquad y \uparrow 1.
\]
For the marginal case $\gamma=3$, the singular part of $Q$ involves $(1-w)^2 \log(1-w)$ rather than a pure power. The inversion then naturally leads to the Lambert $W$-function \cite{Corless1996LambertW}.

\begin{proposition}[Asymptotics and near-closed form of $P(y)$ for a logarithmic $Q$]
Let $P$ solve
\[
P(y) = y\,Q(P(y)),
\]
where, as $\omega \uparrow 1$,
\begin{equation}\label{app4eq2}
\begin{aligned}
    Q(\omega)
=& \omega - \frac{1}{2\,\zeta(2)}(1-\omega)^2\log(1-\omega)
\\&+ o\bigl((1-\omega)^2|\log(1-\omega)|\bigr),
\end{aligned}
\end{equation}
with $\zeta(2)=\pi^2/6$ and natural logarithms. Set $x = 1-y$ and $u = 1-\omega$. Then, as $y \uparrow 1$,
\begin{align}
P(y)
&= 1 - \sqrt{\frac{-\,4\,\zeta(2)\,(1-y)}{W_{-1}\!\bigl(-4\,\zeta(2)\,(1-y)\bigr)}}\,, \label{app4eq3}\\[2mm]
P(y)
&= 1 - 2\sqrt{\zeta(2)}\,
        \sqrt{\frac{1-y}{\lvert \log(1-y)\rvert}}\,
        \bigl(1+o(1)\bigr), \label{app4eq4}
\end{align}
where $W_{-1}$ denotes the $(-1)$ branch of the Lambert $W$ function.
\end{proposition}
\begin{proof}
Write $\omega = P(y)$ and introduce
\[
u = 1-\omega \downarrow 0,
\qquad x = 1-y \downarrow 0.
\]
From \eqref{app4eq2},
\[
Q(\omega)
= (1-u) - a\,u^2 \log u + o\bigl(u^2|\log u|\bigr),
\]
with $ a = \frac{1}{2\zeta(2)}$.
Using $P(y)=y\,Q(P(y))$ in the equivalent form $y=\omega/Q(\omega)$, we obtain
\begin{align}
 y &= \frac{1-u}{(1-u) - a\,u^2\log u + o(u^2|\log u|)},\\
  &= 1 + \frac{a\,u^2\log u}{1-u-a\,u^2\log u}
    + o\bigl(u^2|\log u|\bigr).   
\end{align}
Since $u\to 0$ and $\log u < 0$, the denominator tends to $1$, so
\begin{align}\label{app4eq6}
x &= 1-y,\\ &= -a\,u^2\log u + o\bigl(u^2|\log u|\bigr),\\
 &= a\,u^2|\log u|\,(1+o(1)).
\end{align}
Set $t = u^2 \in (0,1)$. Then $|\log u| = \tfrac12|\log t|$, and \eqref{app4eq6} becomes
\begin{align}
 x& = \frac{a}{2} t\,|\log t|\,(1+o(1)),\\
  &= \frac{1}{4\zeta(2)}\,t\,|\log t|\,(1+o(1)).   
\end{align}
Since $t\in(0,1)$, we have $|\log t| = -\log t$, so
\[
-4\zeta(2)\,x = t \log t\,(1+o(1)).
\]
Ignoring the lower-order $o(1)$ in the inversion (which does not affect the leading behavior nor the Lambert–$W$ representation), we solve
\[
t \log t = c, \qquad c = -4\zeta(2)\,x \in (-e^{-1},0).
\]
The solution is
\[
t = \frac{c}{W_{-1}(c)},
\]
where the branch $W_{-1}$ ensures $t \to 0^+$ as $c \to 0^-$. Thus,
\[
t = \frac{-4\zeta(2)\,x}{W_{-1}\!\bigl(-4\zeta(2)\,x\bigr)},
\qquad
u = \sqrt{t}, \qquad
P(y) = 1-u,
\]
which yields the near-closed form \eqref{app4eq3}.
For the explicit asymptotics, use
\[
W_{-1}(-\varepsilon)
 = \log \varepsilon - \log|\log \varepsilon| + o(1),
\qquad \varepsilon \downarrow 0.
\]
This gives
\begin{align*}
  u^2 &= \frac{4\zeta(2)\,x}{|\log x|}\,(1+o(1)),
\qquad\\  
u& = 2\sqrt{\zeta(2)}\,\sqrt{\frac{x}{|\log x|}}\,(1+o(1)),
\end{align*}

and \eqref{app4eq4} follows upon substituting $x = 1-y$.
\end{proof}

\begin{remark}
Subleading regular terms of order $(1-\omega)^2$ in $Q(\omega)$ only modify the argument of $W_{-1}$ by a constant multiplicative factor (e.g.\ an $e^{-3}$), and do not alter the leading square-root behavior
\[
P(y) \sim 1 - 2\sqrt{\zeta(2)} 
\sqrt{\frac{1-y}{|\log(1-y)|}} \quad (y \uparrow 1),
\]
nor the associated avalanche exponent.
\end{remark}
From the above analysis, one finally obtains the leading behavior of $P(y)$ as $y\uparrow 1$
\begin{equation}
    P(y) \approx
    \begin{cases}
        1 - \displaystyle
        \biggl[\dfrac{\zeta(\gamma-1)}{\Gamma(1-\gamma)}
        (1-y)\biggr]^{\frac{1}{(\gamma-1)}}, & 2 < \gamma < 3, \\[8pt]
        1 - 2\sqrt{\zeta(2)}\,
        \sqrt{\dfrac{1-y}{|\log(1-y)|}}, & \gamma = 3, \\[8pt]
        1 - \displaystyle
        \biggl[\dfrac{2\,\zeta(\gamma-1)}{\zeta(\gamma-2)}\biggr]^{\!1/2}
        (1-y)^{1/2}, & \gamma > 3.
    \end{cases}
\end{equation}
Standard singularity analysis~\cite{flajolet1988singularity} applied to \eqref{SEC1EQ6} then yields the asymptotics of the avalanche size distribution $p(s)$
\begin{equation}
p(s) \sim
\begin{cases}
a(\gamma)\, s^{-\gamma/(\gamma-1)}, & 2 < \gamma < 3, \\[1ex]
b\, s^{-3/2} \bigl[\ln s\bigr]^{-1/2}, & \gamma = 3, \\[1ex]
c(\gamma)\, s^{-3/2}, & \gamma > 3,
\end{cases}
\qquad (s \to \infty),
\end{equation}
where $b=\sqrt{\pi/6}$, $c(\gamma)=\sqrt{\zeta(\gamma-1)/(2\pi\zeta(\gamma-2))}$ and   $a(\gamma)=-((\Gamma(1-\gamma)/\zeta(\gamma-1))^{1/1-\gamma})/(\Gamma(1/(1-\gamma)))$  are positive constants.

\section{Computation of the approximated expression~\ref{SEC3EQ2}}
\label{app5}

In Section~\ref{ssec:SF} we have determined an analytical expression for the avalanche size distribution $p(s)$ given by~\eqref{eq:psSFexact} 
\begin{align*}
p(s)
&=
\frac{(1-f)^{s-1}}{s}
\sum_{m=0}^{s}\binom{s}{m}A^{s-m}B^m\times\\
&\qquad \sum_{k_1,\dots,k_m\ge 1}
\frac{\binom{K}{s-1}f^{K-s+1}}{k_1^\gamma\cdots k_m^\gamma}
\mathbf{1}_{\{K\ge s-1\}}\, ,
\end{align*}
where 
\[
A=1-\frac{\zeta(\gamma)}{\zeta(\gamma-1)}\,\,,\,
B=\frac{1}{\zeta(\gamma-1)}\text{ and }
L(w)=\mathrm{Li}_{\gamma}(u(w)) \, .
\]
In the limit of large $K=k_1+\cdots+k_m$, the latter expression can be approximated by~\eqref{SEC3EQ2}
\begin{equation*}
p(s)
\approx
\frac{1}{\zeta(\gamma-1)}
\frac{(1-f)^{\gamma-1}}{f^\gamma}
s^{-\gamma}\, .
 \end{equation*}

The aim of this Section is to prove the above claim. Let us thus define
\[
C_m(K)=
\sum_{\substack{k_1+\cdots+k_m=K\\ k_i\ge 1}}
\frac{1}{k_1^\gamma\cdots k_m^\gamma}.
\]
We identify the binomial factor
\[
\binom{K}{s-1}(1-f)^{s-1}f^{K-s+1}
=
\mathbb P(X=s-1),
\]
with $X\sim \mathrm{Bin}(K,1-f).
$
For large \(K\), we approximate it by a Poisson law with
\[
\lambda=K(1-f),
\]
so that
\[
\binom{K}{s-1}(1-f)^{s-1}f^{K-s+1}
\approx
e^{-K(1-f)}
\frac{[K(1-f)]^{s-1}}{(s-1)!}.
\]
Using Stirling's formula,
\[
(s-1)!\sim \sqrt{2\pi(s-1)}
\left(\frac{s-1}{e}\right)^{s-1},
\]
we get the Gaussian approximation
\begin{align*}
  e^{-K(1-f)}
\frac{[K(1-f)]^{s-1}}{(s-1)!}
&\approx
\frac{1}{\sqrt{2\pi(s-1)}}\times\\
&\qquad\exp\left[
-\frac{(K(1-f)-(s-1))^2}{2(s-1)}
\right].
\end{align*}

Replacing \(s-1\) by \(s\), this becomes
\[
\frac{1}{\sqrt{2\pi s}}
\exp\left[
-\frac{(K(1-f)-s)^2}{2s}
\right].
\]
Hence
\begin{align*}
  p(s)&\approx
\frac{1}{s}
\sum_{m=0}^{s}
\binom{s}{m}A^{s-m}B^m
\sum_K
C_m(K)\times\\&\qquad
\frac{1}{\sqrt{2\pi s}}
\exp\left[
-\frac{(K(1-f)-s)^2}{2s}
\right].  
\end{align*}
The Gaussian is concentrated around
\[
K_\star=\frac{s}{1-f}.
\]
Therefore,
\begin{align*}
\frac{1}{1-f}
C_m\left(\frac{s}{1-f}\right)&\approx 
  \sum_K C_m(K)
\frac{1}{\sqrt{2\pi s}}\times \\&\qquad
\exp\left[
-\frac{(K(1-f)-s)^2}{2s}
\right].
\end{align*}

Thus
\[
p(s)\approx
\frac{1}{s(1-f)}
\sum_{m=1}^{s}
\binom{s}{m}A^{s-m}B^m
C_m\left(\frac{s}{1-f}\right).
\]
Now we approximate \(C_m(K)\). Since the variables satisfy \(k_i\ge 1\), when one variable is large and the others are small, the large one is shifted:
\[
k_{\mathrm{large}}
\approx
K-\sum_{i=1}^{m-1}k_i.
\]
With
\[
\bar r
=
\frac{\sum_{r\ge 1} r r^{-\gamma}}
{\sum_{r\ge 1} r^{-\gamma}}
=
\frac{\zeta(\gamma-1)}{\zeta(\gamma)},
\]
we use
\[
C_m(K)
\approx
m\zeta(\gamma)^{m-1}
\left(K-(m-1)\bar r\right)^{-\gamma}.
\]
Therefore
\[
C_m\left(\frac{s}{1-f}\right)
\approx
m\zeta(\gamma)^{m-1}
\left(
\frac{s}{1-f}-(m-1)\bar r
\right)^{-\gamma}.
\]
So
\begin{align*}
  p(s)&\approx
\frac{1}{s(1-f)}
\sum_{m=1}^{s}
\binom{s}{m}
A^{s-m}B^m
m\zeta(\gamma)^{m-1}\times\\&\qquad
\left(
\frac{s}{1-f}-(m-1)\bar r
\right)^{-\gamma}.  
\end{align*}
Let
\[
q=B\zeta(\gamma)
=
\frac{\zeta(\gamma)}{\zeta(\gamma-1)}.
\]
Since
\[
A+B\zeta(\gamma)=1,
\]
the weights
\[
\binom{s}{m}A^{s-m}
(B\zeta(\gamma))^m
\]
are binomial weights with
\[
M\sim \mathrm{Bin}(s,q).
\]
Hence the sum can be written approximately as
\[
S
=
\sum_{m=1}^{s}
\binom{s}{m}
A^{s-m}B^m
m\zeta(\gamma)^{m-1}
\left(
\frac{s}{1-f}-(m-1)\bar r
\right)^{-\gamma}.
\]
Equivalently,
\[
S
=
\frac{1}{\zeta(\gamma)}
\mathbb E\left[
M
\left(
\frac{s}{1-f}-(M-1)\bar r
\right)^{-\gamma}
\right].
\]
For large \(s\),
\[
M\approx sq.
\]
Thus
\[
S
\approx
\frac{sq}{\zeta(\gamma)}
\left(
\frac{s}{1-f}-(sq-1)\bar r
\right)^{-\gamma}.
\]
But
\[
\bar r q
=
\frac{\zeta(\gamma-1)}{\zeta(\gamma)}
\frac{\zeta(\gamma)}{\zeta(\gamma-1)}
=1.
\]
Therefore
\[
\frac{s}{1-f}-sq\bar r
=
\frac{s}{1-f}-s
=
s\frac{f}{1-f}.
\]
Neglecting the lower-order \(+\bar r\) correction, we obtain
\[
S
\approx
\frac{sq}{\zeta(\gamma)}
\left(
s\frac{f}{1-f}
\right)^{-\gamma}.
\]
Since
\[
\frac{q}{\zeta(\gamma)}
=
\frac{1}{\zeta(\gamma-1)}
=
B,
\]
we get
\[
S
\approx
B s^{1-\gamma}
\left(\frac{1-f}{f}\right)^\gamma.
\]

Finally,
\[
p(s)
\approx
\frac{1}{s(1-f)}
B s^{1-\gamma}
\left(\frac{1-f}{f}\right)^\gamma.
\]

Hence
\[
p(s)
\approx
\frac{1}{\zeta(\gamma-1)}
\frac{(1-f)^{\gamma-1}}{f^\gamma}
s^{-\gamma}.
\]

In Fig.~\ref{Br1} we report numerical results to support the previous computation and in Fig.~\ref{sp_Sf_All}, we report the transition from $\tau=\gamma/\gamma-1$ to $\tau=\gamma$ for scale-free networks of different sizes with fixed degree exponent $\gamma=3$.

\begin{figure}
    \centering
    \includegraphics[width=0.95\linewidth]{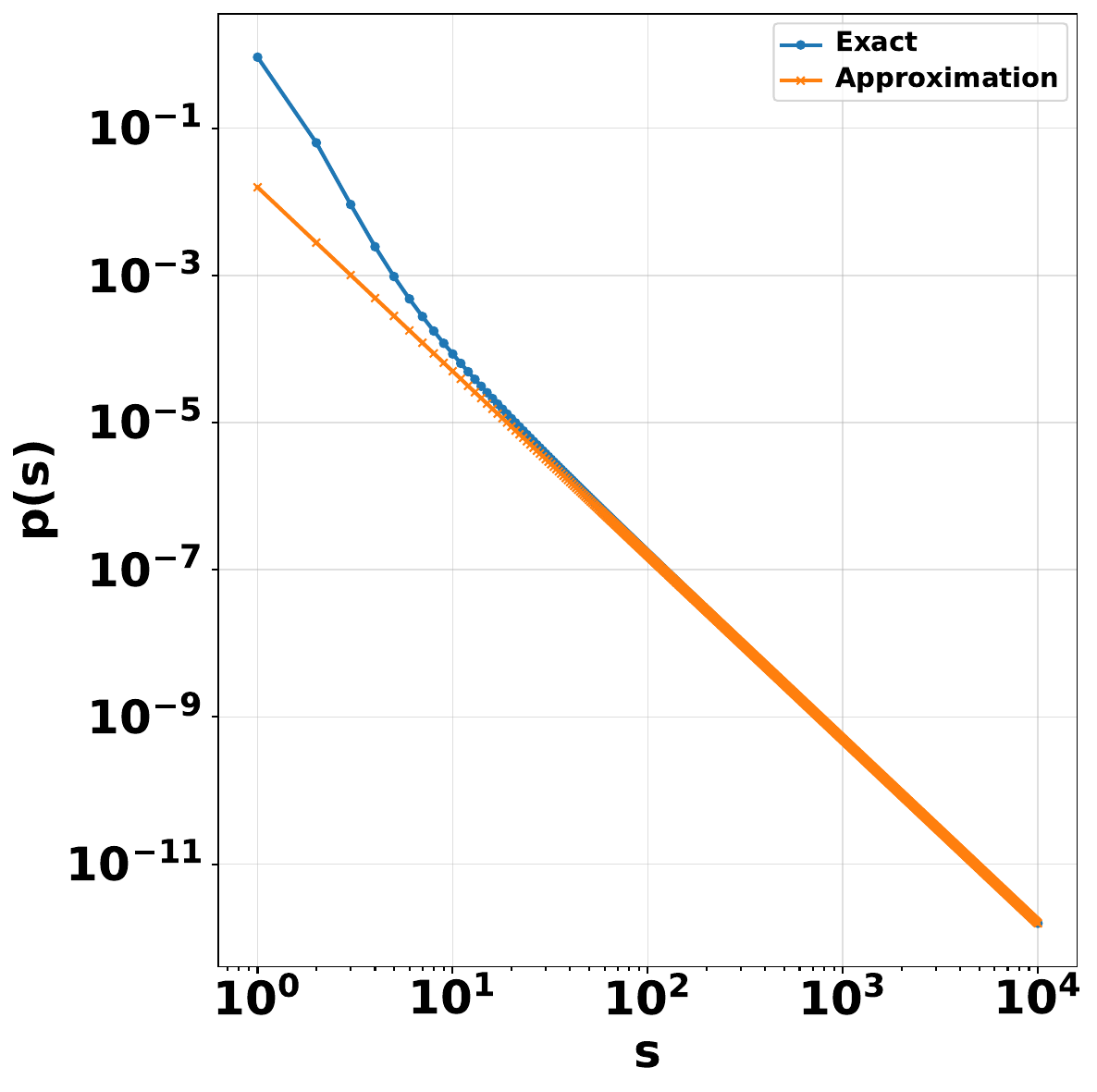}
    \caption{Comparison between the exact distribution $p(s)$ given by~\eqref{eq:psSFexact} and the approximated one given by~\ref{SEC3EQ2} for $\gamma=2.5$ and $f=0.9$.}
    \label{Br1}
\end{figure}

\begin{figure}
    \centering
    \includegraphics[width=0.95\linewidth]{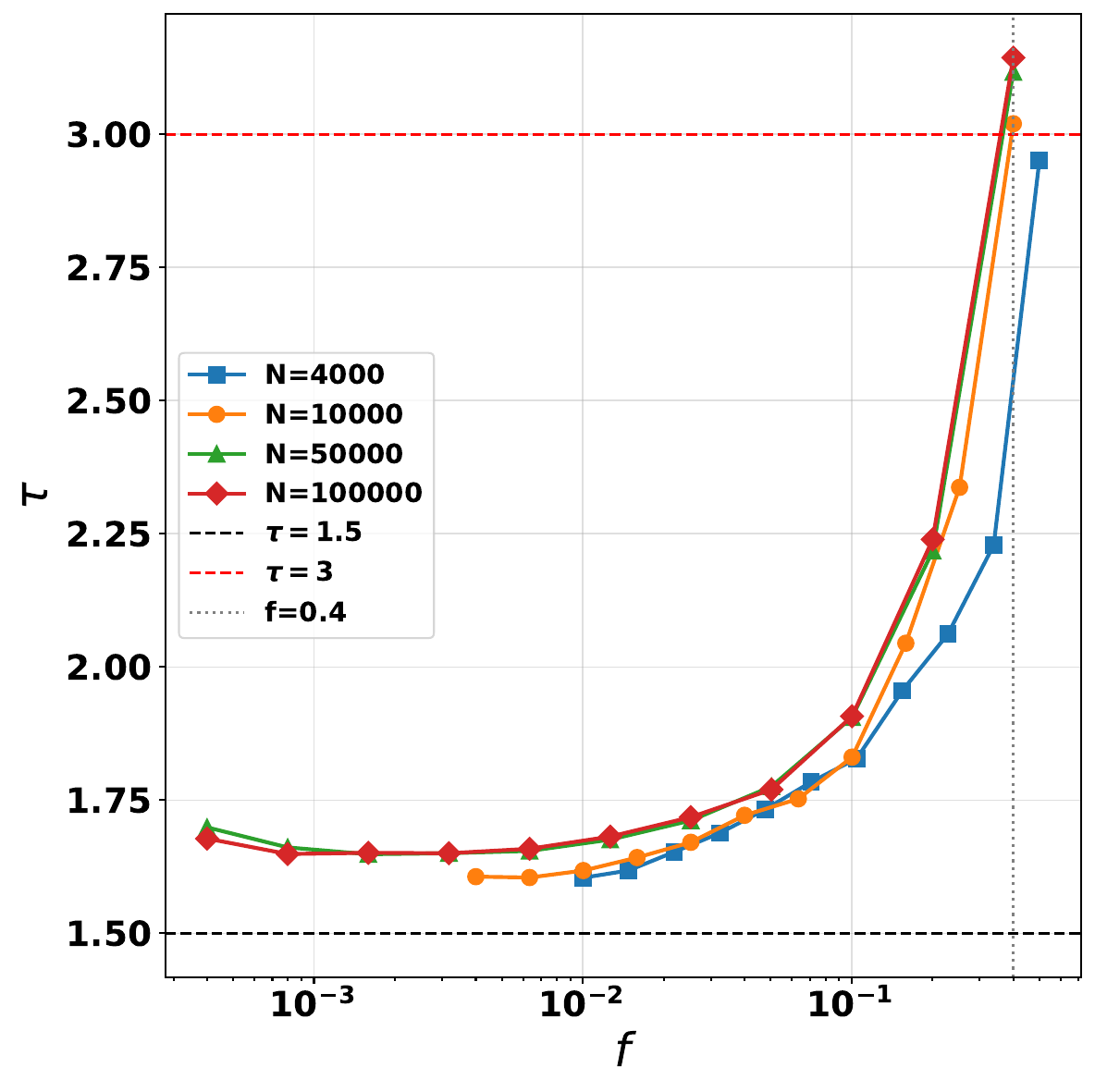}
    \caption{Evolution of the fitted exponent $\tau$ as a function of the dissipation parameter $f$ for scale-free networks of different sizes with fixed degree exponent $\gamma=3$. For small $f$, $\tau$ remains close to the value $\gamma/(\gamma-1)=1.5$, while for larger $f$ it approaches $\gamma=3$. The vertical dotted line indicates the crossover region around $f\simeq 0.4$.
}
    \label{sp_Sf_All}
\end{figure}

\section{Determination of the exact $p(s)$ with the Theorem~\ref{thm:lagrange} in the Erd\H{o}s--R\'enyi Network and random regular cases}\label{app6}

The aim of this Section is to provide a complementary proof of the avalanche size distribution in the case of Erd\H{o}s--R\'enyi and regular random graph, based on Theorem~\ref{thm:lagrange}, and so differing from the one presented in the main text.

In the Erd\H{o}s--R\'enyi case we have,
\[
Q(w|f)=\frac{\lambda-1}{\lambda}+\frac{1}{\lambda}\exp\!\big(\lambda(u(w)-1)\big).
\]
For convenience set
\[
C:=\frac{\lambda-1}{\lambda},\qquad D(w):=\frac{1}{\lambda}\exp\!\big(\lambda(u(w)-1)\big),
\]
so that $Q(w|f)=C+D(w)$.
We can write
\[
Q(w|f)^s=\sum_{t=0}^s \binom{s}{t}\,C^{\,s-t}\,D(w)^t.
\]
We have
\[
D(w)^t=\Big(\frac{1}{\lambda}\Big)^t
\exp\!\big(t\lambda(u(w)-1)\big).
\]
Since $u(w)=(1-f)w+f$,
\[
u(w)-1=(1-f)(w-1),
\]
hence
\begin{align*}
 \exp\!\big(t\lambda(u(w)-1)\big)
&=\exp\!\big(t\lambda(1-f)(w-1)\big)\\
&=e^{-t\lambda(1-f)}\,\exp\!\big(t\lambda(1-f)\,w\big).   
\end{align*}
Therefore
\[
D(w)^t=\Big(\frac{1}{\lambda}\Big)^t e^{-t\lambda(1-f)}\,
\exp\!\big(t\lambda(1-f)\,w\big).
\]
Using $\exp(aw)=\sum_{n\ge0}\frac{a^n}{n!}w^n$, we get
\[
[w^{s-1}]\,\exp\!\big(t\lambda(1-f)\,w\big)=\frac{\big(t\lambda(1-f)\big)^{s-1}}{(s-1)!}.
\]
Hence
\[
[w^{s-1}]\,D(w)^t
=\Big(\frac{1}{\lambda}\Big)^t e^{-t\lambda(1-f)}\,
\frac{\big(t\lambda(1-f)\big)^{s-1}}{(s-1)!}.
\]
By using the fact that $p(s)=\frac1s[w^{s-1}]Q(w|f)^s$, we get
\[
p(s)=\frac1s\sum_{t=0}^s \binom{s}{t}\,C^{\,s-t}\,
\Big(\frac{1}{\lambda}\Big)^t e^{-t\lambda(1-f)}\,
\frac{\big(t\lambda(1-f)\big)^{s-1}}{(s-1)!}.
\]
Since $\frac{1}{s}\cdot\frac{1}{(s-1)!}=\frac{1}{s!}$, this can be written compactly as
\begin{equation*}
 p(s)=\frac{1}{s!}\sum_{t=0}^s \binom{s}{t}\,C^{\,s-t}\,
\Big(\frac{1}{\lambda}\Big)^t e^{-t\lambda(1-f)}\,
\big(t\lambda(1-f)\big)^{s-1},
\end{equation*}
with $ C=\frac{\lambda-1}{\lambda}.$\\

Now in the random regular case,
\[
Q(w|f)=a+b\,u(w)^{k_0},
\qquad
u(w)=(1-f)w+f,\] with
$a=\frac{k_0-1}{k_0}$,
$b=\frac{1}{k_0}$.
Thus
\[
Q(w|f)^s
=
\left(a+bu(w)^{k_0}\right)^s
=
\sum_{m=0}^{s}\binom{s}{m}
a^{s-m}\bigl(bu(w)^{k_0}\bigr)^m.
\]
So
\[
Q(w|f)^s
=
\sum_{m=0}^{s}\binom{s}{m}
a^{s-m}b^m u(w)^{k_0 m}.
\]
\begin{align*}
 u(w)^{k_0 m}
&=
\bigl((1-f)w+f\bigr)^{k_0 m}\\
&=
\sum_{n=0}^{k_0 m}\binom{k_0 m}{n}(1-f)^n f^{k_0 m-n}w^n.   
\end{align*}
Therefore
\[
[w^{s-1}]u(w)^{k_0 m}
=
\binom{k_0 m}{s-1}(1-f)^{s-1}f^{k_0 m-(s-1)},
\]
provided $k_0 m\ge s-1$.
\begin{align*}
p(s)
&=
\frac{1}{s}[w^{s-1}]Q(w|f)^s\\
&=
\frac{1}{s}
\sum_{m=0}^{s}\binom{s}{m}
a^{s-m}b^m
[w^{s-1}]u(w)^{k_0 m}.
\end{align*}
Thus
\begin{align*}
 p(s)
&=
\frac{(1-f)^{s-1}}{s}
\sum_{m=0}^{s}
\binom{s}{m}
a^{s-m}b^m\times\\ &\quad
\binom{k_0 m}{s-1}
f^{k_0 m-s+1}
\mathbf{1}_{\{k_0 m\ge s-1\}}.   
\end{align*}

\section{Goodness-of-fit diagnostics.}
\label{app7}
We recall the goodness-of-fit diagnostics reported in the figure captions.
Let \(C_d\) denote the empirical cumulative distribution function of the
avalanche-size data restricted to the fitting range, and let \(C_t\) denote the
cumulative distribution function of the fitted theoretical model. Both
functions are evaluated on the same set of observed avalanche sizes
\(\{x_i\}_{i=1}^M\). We define
\[
\Delta_i = C_d(x_i)-C_t(x_i),
\qquad i=1,\ldots,M .
\]

The Kolmogorov--Smirnov statistic is
\[
KS = D
=
\max\left\{
\max_{1\le i\le M}\Delta_i,
-\min_{1\le i\le M}\Delta_i
\right\}.
\]
It measures the largest vertical discrepancy between the empirical and fitted
cumulative distributions.

The Kuiper statistic is
\[
V
=
\max_{1\le i\le M}\Delta_i
-
\min_{1\le i\le M}\Delta_i .
\]
Unlike the Kolmogorov--Smirnov distance, it combines the largest positive and
negative deviations and is therefore sensitive to discrepancies over the whole
support.

The Anderson--Darling statistic used here is
\[
A^2
=
\sum_{i=1}^{M}
\frac{\left(C_d(x_i)-C_t(x_i)\right)^2}
{C_t(x_i)\left(1-C_t(x_i)\right)} .
\]
This diagnostic gives more weight to deviations in the tails of the
distribution.

Finally, the normalized coefficient \(\kappa\) is defined by
\[
\kappa
=
1+
\frac{1}{M}
\sum_{i=1}^{M}
\left(C_d(x_i)-C_t(x_i)\right).
\]
Thus,
\[
|\kappa-1|
=
\left|
\frac{1}{M}
\sum_{i=1}^{M}
\left(C_d(x_i)-C_t(x_i)\right)
\right|.
\]
Small values of \(KS\), \(V\), \(A^2\), and \(|\kappa-1|\) indicate a close
agreement between the empirical avalanche-size distribution and the fitted
model. Equivalently, \(\kappa\) close to \(1\) indicates the absence of a
systematic average bias between the empirical and theoretical cumulative
distributions.

\end{document}